\documentclass[10pt,leqno,oneside]{extarticle}

\usepackage{mathptmx} 

\usepackage{titlesec}
\titleformat*{\section}{\large\bfseries} 
\titleformat*{\subsection}{\large\bfseries} 

\usepackage{cite}

\usepackage[small]{caption}
\usepackage{amsmath}
\usepackage{amssymb}
\usepackage{amsfonts}
\usepackage{setspace}
\usepackage{latexsym}
\usepackage{mathrsfs}
\usepackage{epsfig}
\usepackage{amsthm}
\usepackage{yfonts}
\usepackage{graphicx}
\usepackage{color}

\usepackage{MnSymbol}

\usepackage{enumitem}
\setitemize{wide}

\newcommand{\n}{\noindent}
\newcommand{\be}{\begin{equation}}
\newcommand{\ee}{\end{equation}}
\newcommand{\ben}{\begin{displaymath}}

\newcommand{\een}{\end{displaymath}}

\newcommand{\vs}{\vspace{0.2cm}}

\newtheorem{Definition}{Definition}
\newtheorem{Proposition}{Proposition}
\newtheorem{Theorem}{Theorem}

\newtheorem{Lemma}{Lemma}

\newtheorem{AProposition}{Aux-Proposition}

\newtheorem*{T1}{Theorem \ref{T1}}

\usepackage{fancyhdr}
\AtBeginDocument{\thispagestyle{plain}}
\fancypagestyle{plain}{
\fancyhead{}
\fancyfoot{}
\fancyhead[LE,RO]{}
\fancyhead[LO,RE]{}
\fancyfoot[R]{\thepage}

\headsep = 20pt
}

\newcommand{\gcur}{\kappa}
\newcommand{\meanc}{{\rm tr}_{h}\Theta}

\begin{document}

\n {\huge Instability of the extreme Kerr-Newman 

\vs
\n black-holes}

\vspace{.1cm}

\vspace{.3cm}
\n {\sc Martin Reiris}\\
\n {\small email: martin@aei.mpg.de}\\

\n \textsc{Max Planck Institute f\"ur Gravitationsphysik \\ Golm - Germany}\\

\n \begin{minipage}[l]{11cm}
\begin{spacing}{.9}{\small
Using black-hole inequalities and the increase of the horizon's areas, we show that there are arbitrarily small electro-vacuum perturbations of the standard initial data of the extreme Reissner-Nordstr\"om black-hole that, (by contradiction), cannot decay in time into any extreme Kerr-Newman black-hole. This proves the expectation that the family of extreme Kerr-Newman black-holes is unstable. It remains of course to be seen whether the whole family of charged black-holes, including those extremes, is stable or not. }
\end{spacing}
\vspace{.3cm}
{\sc PACS}:\hspace{.1cm} 02.40.Hw,\hspace{.1cm} 02.40.Ma,\hspace{.1cm} 04.20.-q.
\end{minipage}

\vs
\section{Introduction}
In this article it is proved that the family of the so called maximal Kerr-Newman black-holes is unstable. To be concrete it is proved that  
that there are arbitrarily small electro-vacuum perturbations of the standard initial data of the extreme Reissner-Nordstr\"om black-hole that cannot decay in time into any extreme Kerr-Newman black-hole. 

To bring more accuracy to this introduction let us start reviewing the mathematics and the qualitative properties of the extreme black-holes.
The Lorentzian metric of the extreme Kerr-Newman (EKN) space-time  of electric charge $Q_{\rm E}$, magnetic charge $Q_{\rm M}$, angular momentum $J$ and mass $m^{2}=(Q^{2}+\sqrt{4J^{2}+Q^{4}})/2\neq 0$, ($Q^{2}=Q_{\rm E}^{2}+Q_{\rm M}^{2}$), is given by 
\begin{align}\label{EKNM}
\mathbf{g}=&-\frac{\Delta-a^2\sin^2\theta}{\Sigma}dt^2-\frac{2a\sin^2\theta}{\Sigma}(r^2+a^2-\Delta)\, dt\, d\phi \\
\nonumber &+\frac{(r^2+a^2)^2-\Delta a^2\sin^2\theta}{\Sigma}\sin^2\theta\, d\phi^2
+\frac{\Sigma}{\Delta}\, dr^2+\Sigma\, d\theta^2, 
\end{align}
where $a=J/m$, $\Sigma=r^2+a^2\cos^2\theta$, and $\Delta=r^2+a^2+Q^2-2mr$, (see for instance \hspace{-.1mm}\cite{MR0465047}).
The coordinate $t$ ranges in $(-\infty,\infty)$, $r$ in $(m,\infty)$ and $(\theta,\varphi)$ are the standard coordinates of the unit sphere $\mathbb{S}^{2}$. The space-time ${\bf M}$ is therefore diffeomorphic to $\mathbb{R}\times \mathbb{R}\times \mathbb{S}^{2}$.
The electromagnetic potential ${\bf A}$ is given explicitly by 
\begin{equation*}\label{A}
{\bf A}=-\frac{Q_{\rm E}\, r}{\Sigma}\big(\, dt-a\sin^2\theta\, d\phi\, \big)+\frac{Q_{\rm M}\cos\theta}{\Sigma}\big(\, a\, dt - (r^2+a^2)\, d\phi\, \big)
\end{equation*}
and recall that the electromagnetic tensor is ${\bf F}_{ab}=\boldsymbol{\nabla}_{a}{\bf A}_{b}-\boldsymbol{\nabla}_{b} {\bf A}_{a}$, 
\footnote{Note that ${\bf A}$ is not smooth at $\{\theta=0\}\cup \{\theta=\pi\}$. In this article smooth means $C^{\infty}$.}.
The solution is rotational symmetric and stationary. Of particular interest for this article are the EKN solutions with $J=0$, $Q_{\rm M}=0$ but $Q_{\rm E}\neq 0$, which are called extreme Reissner-Nordstr\"om (ERN). When $Q_{\rm E}=1$ the ERN metric (from now on ERN$_{1}$) takes the synthetic form 
\begin{align}\label{RNM}
{\bf g}=-\big(1-1/r\big)^{2}dt^{2}+\frac{1}{\big(1-1/r\big)^{2}}dr^{2}+r^{2}d\Omega^{2}
\end{align}
and the electromagnetic potential simplifies to ${\bf A}=-dt/r$. Over the Cauchy hypersurface $\{t=0\}$ the electric field is $E_{a}={\bf F}_{ab}\, {\bf n}^{b}=\partial_{r}/(r^{2}|\partial_{r}|)$ and the magnetic field is zero, i.e. $B_{a}={\bf \star} {\bf F}_{ab}\, {\bf n}^{b}=0$. Here ${\bf n}$ is the time-like unit normal to $\{t=0\}$. The solution is time symmetric and therefore the second fundamental form $K$ of the slice $\{t=0\}$ is zero. Finally the solution is spherically symmetric and static. For future reference the data set over $\Sigma_{0}:=\{t=0\}$ will be called the {\it standard initial data of the} ERN$_{1}$ {\it solution} and denoted by $(\Sigma_{0};g_{0},K_{0};E_{0},B_{0})$.
\begin{figure}[h]
\centering
\includegraphics[width=8cm,height=6cm]{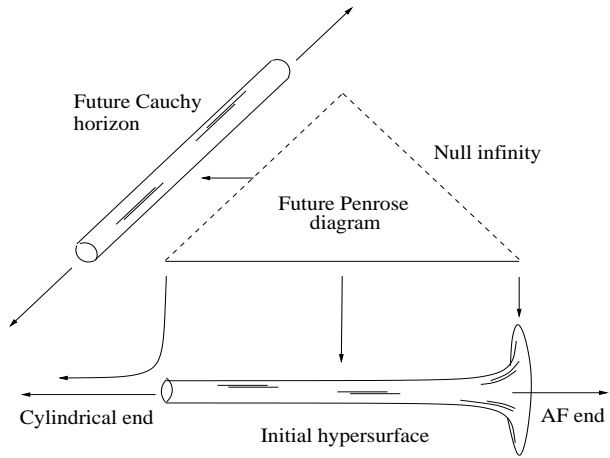}
\caption{Picture of the (half) Penrose diagram of the EKN black-holes. The picture shows also a visualization of the geometry of the standard initial data and the future Cauchy horizon.}
\label{Figure1}
\end{figure}

The EKN solutions form part of the larger family of Kerr-Newman (KN) space-times and lie exactly between those KN space-times representing black-holes and those exhibiting naked singularities.  
Due to their special properties, the EKN solutions have played a peculiar role in the mathematical and physical analysis of black-holes. Some of their most noticeable features are the following. 
The past and the future null infinity of the ERN space-time can be reached from any of its space-time points. Yet the ERN space-time is geodesically incomplete and exhibits future and past Cauchy horizons. Each Cauchy horizon is diffeomorphic to $\mathbb{R}\times \mathbb{S}^{2}$, has complete null generators and the area of any spherical section is
\ben
A=4\pi\sqrt{4|J|+Q^{2}}
\een
In particular, if an extreme solution has $Q_{\rm E}=1$ then to be the one with $Q_{\rm E}=1$, $Q_{\rm M}=0$ and $J=0$ it is necessary and sufficient that  $A=4\pi$.
Moreover the ``initial" Cauchy hypersurface $\{t=0\}$ is maximal and complete (as a Riemannian manifold), and possess no trapped region. This hypersurface is diffeomorphic to $\mathbb{R}\times \mathbb{S}^{2}$ and has one cylindrical end and one asymptotically flat (AF) end (see Figure \ref{Figure1}). 
Of special interest to us is the cylindrical space-time of the ERN$_{1}$ solution (Bertotti's space-time). It is found by taking a sequence $r_{i}\rightarrow 1$, making then the change of variables $\bar{x}=\ln \big((r-1)/(r_{i}-1)\big)$, $\bar{t}=(r_{i}-1)t$ in (\ref{RNM}), and finally taking the limit as $r_{i}\rightarrow 1$. This gives the result
\be\label{RNMT}
\check{\bf g}=-e^{2\bar{x}}d\bar{t}^{2}+d\bar{x}^{2}+d\Omega^{2}
\ee
The three-metric over $\{\bar{t}=0\}$ is then $\check{g}_{0}=d\bar{x}^{2}+d\Omega^{2}$, that is, that of the metric product $\mathbb{R}\times \mathbb{S}^{2}$, hence cylindrical. For future reference, over this slice the electric field is $\check{E}_{0}=\partial_{\bar{x}}$ and the magnetic field $\check{B}_{0}$ and the second fundamental form $\check{K}_{0}$ are zero. The data set $(\mathbb{R}\times \mathbb{S}^{2};\check{g}_{0},\check{K}_{0};\check{E}_{0},\check{B}_{0})$ will be called the {\it standard initial data of the extreme} RN$_{1}$ {\it throat} (ERNT$_{1}$).

It is fundamentally the presence of these peculiar Cauchy horizons what makes extreme solutions so special. Are extreme black-holes physically realistic solutions? Are they stable under small perturbations of the initial data?  What occurs to their horizons under such perturbations? 

A revitalized interest in these old questions reappeared in the last years as a part of new and larger mathematical investigations on the stability of black-hole space-times, \cite{Aretakis:2012ei},\cite{Murata:2013daa},\cite{Dafermos:2013bua},\cite{Bizon:2012we}, \cite{Dain:2012qw} (to mention some). Most of these theoretical developments are characterized by the use of linear techniques over the otherwise unperturbed ERN background.  
As a contribution to the ongoing discussion we prove here that there are arbitrarily small perturbations of the standard ERN$_{1}$ initial data whose evolution cannot decay in any way into any EKN solution. The proof is satisfactory to us in that it is the result of combining black-hole inequalities \cite{Dain:2011mv},\cite{Dain:2011kb}, and the ubiquitous law of area increase of event horizons \cite{Chrusciel:2000cu}, and does not rely in any linear or linearization technique. 
In a sense, our argument belongs to a class of natural procedures to prove instabilities that was used in the literature during the last years \footnote{I would like to thank Piotr Chrusciel for making this remark to me.}
and which consists in finding certain inequalities at the level of the perturbed initial data that are shown to be propagated along the evolution and that are incompatible with the stationary states that one wants to rule out as the long time limit of the evolution (see for instance \cite{Figueras:2011he} and references therein).
\begin{figure}[h]
\centering
\includegraphics[width=10cm,height=4cm]{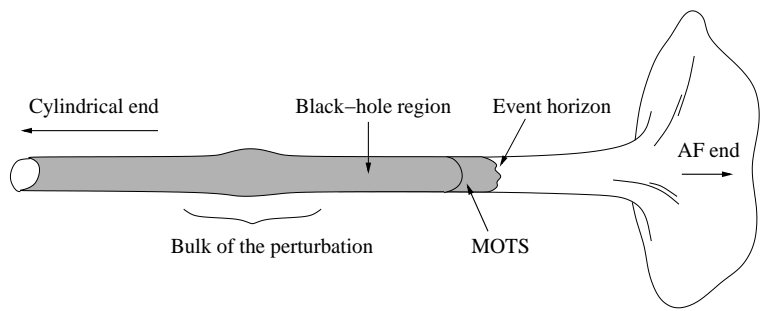}
\caption{Diagram of the initial data used in this article.}
\label{Figure2}
\end{figure}

Before we pass to explain the generalities behind the proof, let us explain in precise terms the main statement to be proved. We first introduce the notion of ``perturbation'' of the standard initial data $(\Sigma_{0};g_{0},K_{0};E_{0},B_{0})$ of the ERN$_{1}$ space-time. 
\begin{Definition}\label{DPER} Let $(\Sigma;g,K;E,B)$ be a smooth an maximal electro-vacuum data set and let $k$ be an integer greater or equal than $1$. We say that the data set is $\varepsilon$-close in $C^{k}$ to the ERN$_{1}$ standard initial data iff there is a diffeomorphism $\varphi: \Sigma_{0}\rightarrow \Sigma$ such that for any $(U, U_{0})$ equal to either $(g,g_{0}),\ (K, K_{0}),\ (E,E_{0})$ or $(B, B_{0})$
we have
\ben
\big\|\, \varphi^{*} U - U_{0}\, \big\|_{C^{k}_{g_{0}}(\Sigma_{0})}\leq \varepsilon.
\een
\end{Definition} 
\n The $C^{k}_{g_{0}}$ norm of a tensor $W$ (no matter its valence) is defined as usual by
\ben
\big\|\, W\, \big\|^{2}_{C^{k}_{g_{0}}(\Sigma_{0})}=\sup_{p\in \Sigma_{0}}\, \bigg[\sum_{j=0}^{j=k}\, \big|\big(\nabla^{(j)}\, W\big) (p)\big|^{2}_{g_{0}}\bigg]
\een
The Definition \ref{DPER} is satisfactory but we need to make sure that the perturbation ``falls off'' along the asymptotically cylindrical end and that the ``cylindrical asymptotic'' is preserved. To be concrete we will work with perturbations that ``fall off exponentially along the cylindrical end into the ERN$_{1}$ standard initial data''. Precisely,
we say that a data set $(\Sigma;g,K;E,B)$, $\varepsilon$-close in $C^{k}$ to $(\Sigma_{0};g_{0},K_{0};E_{0},B_{0})$, {\it falls off exponentially into} $(\Sigma_{0};{g}_{0},{K}_{0};{E}_{0},{B}_{0})$ {\it along the cylindrical end} iff there is $\Lambda>0$ such that for any $(U,U_{0})$ equal to either $(g,{g}_{0}), (K,K_{0}), (E,E_{0})$ or $(B,B_{0})$ we have
\ben
\lim_{r(p)\rightarrow 1}\ e^{\displaystyle \Lambda \ln (r-1)}\, \bigg[\sum_{j=0}^{j=k}\, \big|\big(\nabla^{(j)}\, (\varphi^{*} U -U_{0}) \big) (p)\big|^{2}_{g_{0}}\bigg]=0,
\een  
where $\varphi^{*}$ is the pull-back by the diffeomorphism $\varphi:\Sigma_{0}\rightarrow \Sigma$ (note that $r(p)\rightarrow 1$ means that ``$p$'' diverges along the cylindrical end).

With all these definitions at hand we can state our main result as follows.
\begin{T1} For any $\bar{\varepsilon}>0$ and integer $k\geq 1$ there is a smooth and maximal electro-vacuum data set $(\bar{\Sigma};\bar{g},\bar{K};\bar{E},\bar{B})$, $\bar{\varepsilon}$-close in $C^{k}$ to the standard ERN$_{1}$ initial data and falling into it exponential along the cylindrical end, which cannot decay, towards the future or the past, into any EKN solution.
\end{T1}

Let us overview now the arguments behind the proof. Technical but important information has to be found inside the text. The argument that follows can be done in any time direction. The idea is to construct (arbitrarily small) axisymmetric perturbations of the standard ERN$_{1}$ initial data and do so with sufficiently control to be able to prove that a Marginally Outer Trapped Surface (MOTS) forms separating the two ends (see Figure \ref{Figure2}). In addition, the perturbation is done keeping $Q_{\rm E}=1$, $Q_{\rm M}=0$ and $J=0$. 
In particular, and because the electromagnetic charges and the angular momentum are conserved, if the perturbation evolves into an EKN space-time in the long-time, then it must be one with $Q_{\rm E}=1$, $Q_{\rm M}=0$ and $J=0$, that is, it has to be the ERN that is being perturbed \footnote{To be certain here, the charges and the angular momentum are not only conserved at null infinity, they take also the same values over any embedded sphere isotopic to a ``sphere'' at ``spatial infinity''. This is explained in Section \ref{BACKM}.}.
Moreover, due to presence of a MOTS which acts as a barrier,  the event horizon must intersect the initial Cauchy hypersurface somewhere between the MOTS and the asymptotically flat end. In parallel to all this it is shown that every surface $S$ embedded in the initial hypersurface and separating the two ends has area strictly greater than $4\pi$. In particular the intersection of the event horizon and the initial hypersurface must have area strictly greater than $4\pi$. As the areas of sections of the event horizon are non-decreasing in time, we conclude that the initial data cannot evolve into the ERN$_{1}$ solution because its horizon has area exactly $4\pi$. The perturbed data set is depicted in Figure \ref{Figure2} and the (presumed) evolution in Figure \ref{Figure3}.  

Like any argument by contradiction, the one before does not say what indeed occurs during the time evolution. It just says something of what cannot happen. Nevertheless, the presence of the mentioned MOTS in the perturbed initial data suggests that it must decay in the long-time into a non-extremal KN black-hole. For this reason it is expected also that whatever occurs to the ``old'' horizon of the ERN$_{1}$, that part of the space-time stays hidden inside the new black-hole region.  
Regardless of that, this work doesn't yield any light about the fate of the ERN horizon under perturbations. In this sense it doesn't make previous investigations about the ERN horizon less interesting.

In principle, with further work but following a similar argument, one should be able to prove that there are arbitrarily small perturbations of any EKN that cannot decay in any way into an EKN black-hole. What makes the use of the ERN and not of any other EKN solution more useful is that the perturbations can be made time-symmetric and for this reason proving the existence of a MOTS reduces to proving the existence of a minimal surface which is technically more accessible \footnote{I would like to thank Sergio Dain for pointing this out.}. 
 
The organization of this article is the following. In Section \ref{BACKM} we recall the basic material to be used about electro-vacuum space-times. In Section \ref{IMID} we discuss black-hole inequalities on data sets that we call of the ERN$_{1}$ ``type'' and that are introduced in Definition \ref{DERNT}. Roughly speaking, such data sets are defined to share the topology and the asymptotic geometry of the standard initial data of the ERN$_{1}$ solution. Not surprisingly the perturbations of the standard initial data of the ERN$_{1}$ solution that we are going to use are of the ERN$_{1}$ type. The main result of this section is to prove that the area of any (compact, boundaryless and embedded) surface separating the two ends of any data set of the ERN$_{1}$ type is strictly greater than $4\pi$.  The analysis in this section shares many elements with \cite{ReirisZTBH}. In Section \ref{SECDS} we construct the mentioned initial perturbations using the conformal method. The existence of solutions of the conformal equations is proved following standard barrier methods \cite{Chrusciel-Mazzeo} which give good control on the solutions. In Section \ref{RIGS} we show the rigidity of the ERNT$_{1}$ initial data which will be necessary in Section \ref{SECDSU} to show that one can make arbitrarily small perturbations containing MOTS. It is worth mentioning that the rigidity of the ERNT$_{1}$ initial data is of interest in interest. In particular the formation of extreme RN throats along sequence of data sets can be studied in the same way as was done in \cite{ReirisZTBH} with the formation of extreme Kerr-throats. 
The proof of the main result following the lines explained above is made formally and finally in Section \ref{PMR}.
\section{Background material}\label{BACKM} 

In this section we recall succinctly and with certain formality those notions, like that of electric and magnetic charges, that will be necessary throughout the article. The formal treatment is justified by the mathematical nature of the paper. 
 
We will be working with smooth electro-vacuum space-times $({\bf M};{\bf g};{\bf F})$, where $({\bf M};{\bf g})$ an orientable and time orientable Lorentzian manifold. We will assume that an orientation on ${\bf M}$ was chosen and that a future direction was assigned. Let $\Sigma$ be a space-like hyper-surface and ${\bf n}$ a future unit normal to $\Sigma$. As usual, the orientation on ${\bf M}$ and the field ${\bf n}$ provide an orientation on $\Sigma$, more precisely: $\{e_{1}(p),e_{2}(p),e_{3}(p)\}$ is a positive basis of $T_{p}\Sigma$ iff $\{{\bf n}(p),e_{1}(p),e_{2}(p),e_{3}(p)\}$ is a positive basis of $T_{p}{\bf M}$.
Space-times tensors, like the Ricci curvature ${\bf Ric}$ of ${\bf g}$, will be boldfaced.  

\vs
\n {\bf (i) The Einstein-Maxwell system.}

\vs
\n In coordinate-independent form the Einstein-Maxwell equations are
\be\label{EME}
{\bf Ric}-\frac{1}{\displaystyle 2}\, {\bf R}\, {\bf g}=8\pi {\bf T},\quad {\rm d}\, {\bf F}=0,\quad \text{and}\quad {\rm d}\, {\bf \star}\, {\bf F}=0
\ee
where ${\rm d}$ is the exterior derivative and $\star$ is the ${\bf g}$-Hodge star, namely ${\bf \star} {\bf F}_{ab}={\bf \epsilon}_{abcd}{\bf F}^{cd}/2$. The electromagnetic energy-momentum tensor ${\bf T}$ appearing in (\ref{EME}) is
\ben
{\bf T}_{ab}=\frac{1}{4\pi}\big({\bf F}_{ac}{\bf F}_{b}^{\ c}-\frac{1}{4}{\bf F}_{cd}{\bf F}^{cd} {\bf g}_{ab}\big).
\een

The 3+1 picture of (\ref{EME}) will be also used during the article. We recall it in what follows \cite{MR583716}. Let $\Sigma_{0}$ be a space-like hyper-surface (possibly with boundary) and ${\bf V}$ a nowhere zero time-like vector field defined on an open neighborhood of $\Sigma_{0}$. 
By moving $\Sigma_{0}$ along ${\bf V}$ one obtains a flow of space-like hypersurfaces $\Sigma_{t}$ (at least for a short time). Coordinates charts $(x^{1},x^{2},x^{3})$ are propagated by ${\bf V}$ to every $\Sigma_{t}$ and any two $\Sigma_{t}$ and $\Sigma_{t'}$ are naturally diffeomorphic. In this way one obtains a flow $(g_{ij}(t),K_{ij}(t))$ of induced three-metrics and second fundamental forms on the fixed manifold $\Sigma_{0}$. Writing ${\bf V}|_{\Sigma_{t}}=N(t){\bf n}+X^{i}(t)\partial_{i}$, where ${\bf n}$ is a future unit normal to $\Sigma_{t}$, one obtains also a flow of lapse functions $N(t)$ and shift vectors $X(t)=X^{i}\partial_{i}$. In this $3+1$ setup the Einstein equation (first eq. in (\ref{EME})) is    
\ben
\left\{
\begin{array}{l}
\dot{g}_{ij}=-2NK_{ij} +{\mathcal L}_{X} g_{ij},\vs \\
\dot{K}_{ij}=-\nabla_{i}\nabla_{j} N +N(Ric_{ij}-2K_{il}K^{l}_{\ j}) + {\mathcal L}_{X} K_{ij} - 8\pi N({\bf T}_{ij}+\frac{1}{2} ({\bf T}_{ab}{\bf g}^{ab})g_{ij}),\vs \\
R=|K|^{2}-k^{2}+16\pi {\bf T}_{00},\vs\\
\nabla^{i}K_{ij}-\nabla_{j}k=8\pi {\bf T}_{0i},
\end{array}
\right.
\een
where ${\bf T}_{00}={\bf T}({\bf n},{\bf n})$ and ${\bf T}_{0i}={\bf T}({\bf n},\partial_{i})$, $\nabla$ is the $g$-covariant derivative, $k={\rm tr}_{g} K$ is the mean curvature and ${\mathcal L}$ is the Lie-derivative. The space-time metric is written in the form 
\ben
{\bf g}=-(N^{2}-X_{i}X^{i})dt^{2}+X_{i}(dt\otimes dx^{i}+dx^{i}\otimes dt)+g_{ij}dx^{i}dx^{j}
\een 
At every slice $\Sigma_{t}$, the electric and magnetic fields $E$ and $B$, are defined by $E^{i}={\bf F}^{i}_{\ a}{\bf n}^{a}$ and $B^{i}={\bf \star F}^{i}_{\ a}{\bf n}^{a}$. In terms of them the electro-vacuum constraint equations are  
\be\label{CEEM}
\left\{
\begin{array}{l}
R=|K|^{2}-k^{2}+2\big(|E|^{2}+|B|^{2}\big),\vs\\
\nabla^{i}K_{ij}-\nabla_{j}k=2(E\times B)_{j},\vs\\
\nabla^{i} E_{i}=0,\vs \\ 
\nabla^{i}B_{i}=0
\end{array}
\right.
\ee
where $(E\times B)_{j}=\epsilon_{ijk}E^{j}B^{k}$. A set $(g,K;E,B)$ satisfying the constraint equations (\ref{CEEM}) on a manifold $\Sigma$ is called an {\it electro-vacuum data set}. The data is maximal if $k={\rm tr}_{g} K=0$.  
 
\vs
\n {\bf (ii) The electric and magnetic charges.} 

\vs
\n Let $[S]$ be an oriented, compact and boundaryless surface $S$ embedded in ${\bf M}$. The bracket $[\ ]$ signifies that an orientation on $S$ has been assigned. Then $Q_{\rm E}([S])$ and $Q_{\rm M}([S])$ are defined by
\ben
Q_{\rm E}([S])=-\frac{1}{4\pi} \int_{[S]} {\bf \star F}\quad\text{and}\quad Q_{\rm M}([S]):=-\frac{1}{4\pi} \int_{[S]} {\bf F}
\een
As ${\rm d}\, {\bf F}=0$ and ${\rm d}\, {\bf \star F}=0$ then $Q_{\rm E}([S])$ and $Q_{\rm M}([S])$ depend only on the homology class of $[S]$. We will be referring this fact as the {\it conservation of charge}.  
If $S$ is embedded in a space-like hypersurface $\Sigma$ then $Q_{\rm E}([S])$ and $Q_{\rm M}([S])$ take the more familiar expressions  
\be\label{ECH}
Q_{\rm E}([S]):=\frac{1}{4\pi}\int_{S} <E,\zeta>\, dA\quad \text{and}\quad Q_{\rm M}([S]):=\frac{1}{4\pi}\int_{S} <B,\zeta>\, dA
\ee
where $<E,\zeta>=E^{i}\zeta^{j}g_{ij}$ and where $\zeta$ the unit normal field to $S$ in $\Sigma$ such that if $\{e_{2}(p),e_{3}(p)\}$ is a positive basis for $T_{p} S$ then $\{n(p),\zeta(p),e_{2}(p),e_{3}(p)\}$ is a positive basis for ${\bf M}$. 
Observe that if $[S]$ and $[S']$ are homologous in $\Sigma$ (and therefore in ${\bf M}$) then the conservations $Q_{\rm E}([S])=Q_{\rm E}([S'])$ and $Q_{\rm M}([S])=Q_{\rm M}([S'])$ can be seen also as a consequence of the laws ${\rm div}\, E=0$ and ${\rm div}\, B=0$ (${\rm div}\, U=\nabla^{i} U_{i}$). 

In this context, the total charges $Q_{\rm E}$ and $Q_{\rm M}$ that show up in the metric expression (\ref{EKNM}) of the EKN solutions are of course the electric and magnetic charges of any sphere with $t$ and $r$ constant and oriented using the outgoing normal $\zeta=\partial_{r}/|\partial_{r}|$ \footnote{Assume $\{\partial_{t},\partial_{r},\partial_{\theta},\partial_{\varphi})$ is positive for ${\bf M}$.}. 

It is the case that the normal $\zeta$ will be given from the context (or simply will not matter). For this reason we will often write $Q_{\rm E}(S)$ and $Q_{\rm M}(S)$. 

\vs
\n {\bf (iii) Angular momentum in electro-vacuum space-times.} 

\vs
\n Suppose now that the electro-vacuum space-time $({\bf M};{\bf g};{\bf F})$ is axisymmetric and that ${\bf F}={\rm d} {\bf A}$ with the potential ${\bf A}$ axisymmetric \footnote{If ${\bf F}$ is exact then an axisymmetric potential $A$ can always be found by averaging any potential by the rotational group $U(1)$. Observe too that ${\bf F}$ is exact iff all the magnetic charges (i.e. $Q_{\rm M}([S])=0$ for all $S$) are zero.}. Denote by $\xi$ the axisymmetric Killing field. Then the angular momentum of an oriented and axisymmetric (compact and boundaryless) surface $[S]$ is \cite{MR0465047}
\be\label{ANGMOM}
J([S]):=\frac{1}{8\pi}\int_{[S]} {\bf \star} (\boldsymbol{\nabla}_{a}\xi_{b}) + \frac{1}{4\pi}\int_{[S]} ({\bf A}^{a}\xi_{a}) {\bf \star F}
\ee  
The angular momentum is conserved too \cite{MR0465047}. Namely if $[\Sigma]$ is an oriented compact and axisymmetric hypersurface of ${\bf M}$ and $\partial [\Sigma]=[S]-[S']$ then $J([S])=J([S'])$. 

If $S$ is embedded in an axisymmetric Cauchy hypersurface $\Sigma$, then the first term in (\ref{ANGMOM}) (which is the Komar angular momentum) reduces to the standard form $(\int_{S}K(\xi,\zeta)dA)/8\pi$ and is therefore zero when $K=0$. If in addition $B=0$ over $\Sigma$ then the second term in (\ref{ANGMOM}) is also zero. To see this use the axisymmetry of ${\bf A}$ to get $\xi^{a}{\bf F}_{ai}=\nabla_{i} {\bf A}(\xi)$ and to conclude that ${\bf A}(\xi)$ must be a constant over $\Sigma$. When $S$ is in addition a sphere then the constant must be zero because ${\bf A}(\xi)$ must vanish at the axes. This information shows that the perturbations constructed in Section \ref{SECDS}, which have $K=0$ and $B=0$, also have total angular momentum $J$ equal to zero.  

\vs
\n {\bf (iv) The stability inequality of minimal surfaces embedded in maximal data sets.} 

\vs
\n Let $(\Sigma;g,K;E,B)$ be an electro-vacuum data set and suppose that $S$ is a (compact, boundary-less and orientable) minimal surface embedded in $\Sigma$. Recall that a surface $S$ is said minimal inside $(\Sigma;g)$ if its mean curvature is identically zero. Let $\zeta$ be a unit normal vector field to $S$ in $\Sigma$ and let  $\alpha:S\rightarrow \mathbb{R}$ be a smooth function. 
The first variation of area when $S$ is deformed along $\alpha\zeta$ is zero by minimality. Instead, the second variation is \cite{MR2780140}
\be\label{SIM}
A''_{\alpha}(S):=\int_{S}\big[\, |\nabla \alpha|^{2}-\big(|\Theta|^{2}+Ric(\zeta,\zeta)\big)\,\alpha^{2}\big]\, dA,
\ee  
where here $\Theta$ is the second fundamental form of $S$. The surface $S$ is said to be stable if $A''_{\alpha}(S)\geq 0$ for all $\alpha$. In 
dimension three the r.h.s of (\ref{SIM}) is simplified due to the identity $2\kappa =(\meanc)^{2}-|\Theta|^{2}+R-2Ric(\varsigma,\varsigma)$, where $\kappa$ is the Gaussian curvature of $S$ (with its induced metric). Using this expression, the minimality of $S$ (i.e. $\meanc=0$) and the energy constraint we deduce that if $S$ is stable then for any $\alpha$ we have
\be\label{SIMD} 
\int_{S} \big(\, |\nabla \alpha|^{2}+\kappa \alpha^{2}\, \big)\, dA\geq  \frac{1}{2} \int_{S} \big( 2|E|^{2}+2|B|^{2}+|K|^{2}+|\Theta|^{2}-k^{2} \big)\, \alpha^{{2}}\, dA.
\ee 

\section{Black-holes inequalities in maximal data sets} \label{IMID}

\begin{Definition} We say that a sphere $S$ embedded in a maximal electro-vacuum data set $(\Sigma; g,K;E,B)$ is a (normalized) extreme RN sphere if over $S$ we have
\be\label{ERNE}
\gcur=1,\qquad E= \zeta,\qquad B=0,\qquad \Theta=0,\qquad\text{and}\qquad K=0,
\ee
where $\gcur$ is the Gaussian curvature, $\zeta$ is a unit normal to $S$ in $\Sigma$ and $\Theta$ is the second fundamental form of $S$ in $(\Sigma;g)$.
\end{Definition}

Normalized extreme RN spheres $S$ are totally geodesic and have $|Q_{\rm E}(S)|=1$, $Q_{\rm M}(S)=0$ and $A(S)=4\pi$.

The following lemma discusses the equality case in the general inequality $A\geq 4\pi Q_{\rm E}^{2}$ and that was not treated in \cite{Dain:2011kb}. 

\begin{Lemma}\label{RNHG}
Let $S$ be a stable (compact, boundaryless and orientable) minimal surface embedded in a maximal electro-vacuum data set and having $A(S)=4\pi$ and $|Q_{\rm E}(S)|=1$. Then, $S$ is a (normalized) extreme RN sphere.
\end{Lemma}
\begin{proof}[\bf Proof.] Recall from (\ref{SIMD}) that the stability inequality of the area implies
\be\label{SIN}
\int_{S} \big(|\nabla \alpha|^{2} +\gcur \alpha^{2}\big)\, dA\geq \int_{S} \bigg(|E|^{2}+|B|^{2}+\frac{|K|^{2}}{2}+\frac{|\Theta|^{2}}{2}\bigg)\, \alpha^{2}\, dA
\ee
for all $\alpha:S\rightarrow {\mathbb R}$. As $|Q_{\rm E}(S)|=1$ we can select the unit normal field $\zeta$ to $S$ such that $\frac{1}{4\pi}\int_{S} <E,\zeta>\, dA=|Q_{\rm E}(S)|=1$. Choosing $\alpha=1$ in (\ref{SIN}) and using then Gauss-Bonet and that 
\be\label{EESI}
1=|Q_{\rm E}(S)|=\frac{1}{4\pi}\big|\int_{S} <E,\zeta>\, dA\big|\leq \frac{1}{(4\pi)^{1/2}}\bigg(\int_{S} |E|^{2}\, dA\bigg)^{1/2}
\ee
we obtain
\ben
4\pi \geq 4\pi + \int_{S} \bigg(|B|^{2}+\frac{|K|^{2}}{2}+\frac{|\Theta|^{2}}{2}\bigg)\, dA
\een
This shows that $B=0$, $K=0$ and $\Theta=0$ and that equality must hold. Therefore equality must hold also in (\ref{EESI}) which implies (by Cauchy-Schwarz) that $E= \zeta$. It remains to see that $\gcur=1$, i.e. that $S$ has a round metric. Let us show this below.

Using $B=0$, $K=0$, $\Theta=0$ and $E=\zeta$ in (\ref{SIN}) we obtain
\ben
\int_{S} \big(\, |\nabla \alpha|^{2} +(\gcur -1) \alpha^{2}\, \big)\, dA\geq 0
\een
for all functions $\alpha$. This implies that the first eigenvalue $\lambda$ of the operator $\alpha\rightarrow -\Delta\alpha+(\gcur-1)\alpha$ must be non-negative. Denote by $\alpha_{\lambda}$ its eigenfunction (which is unique up to a constant and that is well known to be nowhere zero). Then we have
\be\label{EEQ}
-\Delta\alpha_{\lambda} + (\gcur-1)\alpha_{\lambda}=\lambda \alpha_{\lambda}
\ee
Multiplying by $1/\alpha_{\lambda}$ and integrating over $S$ we obtain
\ben
-\int_{S}|\nabla \ln \alpha_{\lambda}|^{2}\, dA=4\pi \lambda\geq 0
\een
This implies that $\lambda=0$ and that $\alpha_{\lambda}$ is a constant. Using this information in (\ref{EEQ}) we obtain $\gcur=1$ as wished. 
\end{proof}

\begin{Definition}\label{DERNT} A maximal electro-vacuum data set $(\Sigma; g,K;E,B)$ is said to be of the ERN$_{1}$ type if there is a (smooth) diffeomorphism $\varphi:\Sigma_{0}\rightarrow \Sigma$ such that 
\ben
\lim_{p\rightarrow {\mathcal End}} \big|(\varphi^{*} U)(p) - U_{0}(p)\big|_{g_{0}}=0
\een
where $(U,U_{0})$ is any of the pairs $(g,g_{0})$, $(K,K_{0})$, $(E,E_{0})$, $(B,B_{0})$ and $p\rightarrow {\mathcal End}$ means ``as $p$ diverges along the cylindrical end or the asymptotically flat end''. 
\end{Definition}

Observe that we require that $(\varphi^{*}g,\varphi^{*}K;\varphi^{*}E,\varphi^{*}B)$ converges to $(g_{0},K_{0};E_{0},B_{0})$ along the ends only in $C^{0}$. For this reason the ADM masses of both data sets are not necessarily equal. However the total electric and magnetic charges must stay the same as they can be calculated from the formulas (\ref{ECH}) along the divergent sequence of spheres $S_{r_{i}}=\{r=r_{i}\}$ on the cylindrical end. That is, any data set of the ERN$_{1}$ type has total charges $|Q_{\rm E}|=1$ and $Q_{\rm M}=0$.

The next proposition is essentially a particular case of the results in \cite{Dain:2011kb}. We include a proof for a more convenient exposition.

\begin{Proposition}\label{PASO} Let $(\Sigma;g,K;E,B)$ be a maximal electro-vacuum data set of ERN$_{1}$ type. Then every (compact, boundaryless and orientable) embedded surface $S$ which is non-contractible inside $\Sigma$ has 
\be\label{UAI}
|Q_{\rm E}(S)|=1\quad \text{and}\quad A(S)\geq 4\pi.
\ee
\end{Proposition}

\begin{proof}[\bf Proof.] We prove first that $|Q_{\rm E}(S)|=1$. Think $\Sigma$ as $\mathbb{R}^{3}\setminus \{o\}$ and $S$ as a surface embedded in it.  Then recall that any compact, boundary-less and orientable surface embedded in $\mathbb{R}^{3}$ divides $\mathbb{R}^{3}$ into two connected components one of which is necessarily unbounded. As $S$ is non-contractible inside $\mathbb{R}^{3}\setminus \{o\}$ then the bounded component of $\mathbb{R}^{3}\setminus S$ must contain the origin $o$. That is, $S$ separates the two ends of $\Sigma$ and the electric charge of $S$ (with an appropriate normal) must be that of the asymptotically flat end, i.e. $|Q_{\rm E}(S)|=1$. 

We prove now that $A(S)\geq 4\pi$. Assume by contradiction the existence of an $S$ with $A(S)<4\pi$. Let $\underline{A}(S)=\inf \{A(S'),S' \text{\ isotopic\ to\ } S\}$. Then obviously we have $4\pi>\underline{A}$. We claim that we also have $\underline{A}(S)>0$. In fact, if there is a sequence $S'_{j}$ of surfaces isotopic to $S$ such that $A(S'_{j})\rightarrow 0$ then 
\ben
1=|Q_{\rm E}(S'_{j})|=\frac{1}{4\pi} \bigg| \int_{S'_{j}} <E,\zeta>\, dA\bigg| \leq \frac{1}{4\pi} \|\, E\, \|_{L^{\infty}_{g}} A(S'_{j})\rightarrow 0
\een
which would show a contradiction. 

Now,  following \cite{MR678484} ({\sc Theorem 1'} \footnote{There is a caveat here. Strictly speaking {\sc Theorem 1'} applies to manifolds with convex boundary which is not the case here (instead we have an AF end and a Cylindrical end $\sim \mathbb{R}\times \mathbb{S}^{2}$). To apply {\sc Theorem 1'} one can work between two spheres, one convex and far away in the AF end and another far away on the cylindrical end where in a neighborhood of it one modifies slightly the metric to have also a convex boundary. Apply {\sc Theorem 1'} and then show that the minimizer does not intersect the deformed region. The reader can see how this type of argument works when we use as similar one in the proof of Aux-Proposition \ref{AUXP3}.}) 
there is a (non-empty) set of compact boundary-less and non-contractible (inside $\Sigma$) minimal surfaces $\{S_{1},\ldots,S_{l}\}$ embedded in $\Sigma$ and a set of positive integers $\{n_{1},\ldots,n_{l}\}$ such that
\ben
\underline{A}(S)=\sum_{i=1}^{i=l} n_{i}A(S_{i})
\een
As $\Sigma$ is diffeomorphic to $\mathbb{R}^{3}\setminus \{o\}$ then all the $S_{i}$'s must be orientable and therefore stable minimal surfaces \cite{MR678484}. Consider now $S_{1}$ and note that $A(S_{1})\leq \underline{A}(S)<4\pi$. We show now that in addition to this it must also be $A(S_{1})\geq 4\pi$, which is a contradiction. To show $A(S_{1})\geq 4\pi$ we recall (as was shown before) that $|Q_{\rm E}(S_{1})|=1$. Therefore plugging $\alpha=1$ in (\ref{SIN}) we have
\begin{align}
4\pi  \geq \int_{S_{1}} |E|^{2}\, dA\geq \frac{1}{A(S_{1})}\bigg(\int_{S_{1}} |<E,\zeta>|\, dA\bigg)^{2}
 \geq \frac{(4\pi |Q_{\rm E}(S_{1})|)^{2}}{A(S_{1})}=\frac{(4\pi)^{2}}{A(S_{1})}
\end{align}
as wished. \end{proof}

The following crucial refinement of Proposition \ref{PASO} shows that equality in the second equation of (\ref{UAI}) cannot be achieved. The proof is based in similar argument to those in \cite{ReirisZTBH}.

\begin{Proposition}\label{PAS} Let $(\Sigma;g,K;E,B)$ be a maximal electro-vacuum data set of ERN$_{1}$ type. Then every (compact, boundary-less and orientable) embedded surface $S$ which is non-contractible inside $\Sigma$ has 
\ben
\quad A(S)>4\pi.
\een 
\end{Proposition}

\begin{proof}[\bf Proof.] By Proposition \ref{PASO} it is enough to show that equality in (\ref{UAI}) cannot be achieved. Proceeding by contradiction assume then that there is $S_{0}$ with $A(S_{0})=4\pi$. Then observe that if $S$ is isotopic to $S_{0}$ then $S$ is also non contractible inside $\Sigma$. Therefore, again by Proposition \ref{PASO}, we have $A(S)\geq 4\pi$ for any surface $S$ isotopic to $S_{0}$. This implies that $S_{0}$ is minimal and stable \footnote{\label{FOTN1} More explicitly, for any smooth $F:[-\varepsilon,\varepsilon]\times S_{0}\rightarrow \Sigma$ with $F(0,-)={\rm Id} (-)$ and $\varepsilon$ small to have $F(x,-):S_{0}\rightarrow \Sigma$ a smooth embedding, the real function $\lambda\rightarrow A(F(\lambda,S_{0}))$, (which is greater or equal than $4\pi$ for all $\lambda$), must have an absolute minimum at $\lambda=0$. It follows that the first $\lambda$-derivative is zero and the second is non-negative. As this is valid for all $F$ then the surface is minimal and stable.}. By Lemma \ref{RNHG} $S_{0}$ is an extreme RN sphere.  

Let ${\mathscr S}_{0}$ be a large and strictly convex sphere (w.r.t the outer normal) over the asymptotically flat end. Denote by $\Omega_{0}$ the region enclosed by it and the cylindrical end and assume that $S_{0}\subset {\rm Int}(\Omega_{0})$. In what follows we are going to use this region $\Omega_{0}$ together with a positive solution $N=N_{0}$ of 
\be\label{LAP}
\Delta N - |E|^{2}N=0
\ee
over $\Omega_{0}$, asymptotically vanishing over the cylindrical end and not-identical to a constant over $S_{0}$. The existence of such $N_{0}$ is proved as follows. Take any two linearly independent smooth positive functions $f_{1}$ and $f_{2}$ over ${\mathscr S}_{0}$. For $i=1,2$, let $\tilde{N}_{i}$ be the solution to (\ref{LAP}) on $\Omega_{0}$ with the boundary condition $\tilde{N}_{i}|_{{\mathscr S}_{0}}=f_{i}$ and asymptotically vanishing over the cylindrical end of $\Omega_{0}$. By the maximum principle we have $\tilde{N}_{i}>0$ for $i=1,2$. If both solutions are constant over $S_{0}$ then one can take a linear combination $\tilde{N}:=\alpha_{1}\tilde{N}_{1}+\alpha_{2}\tilde{N}_{2}$ vanishing exactly over $S_{0}$ but with $\alpha_{1}\neq 0$ and $\alpha_{2}\neq 0$. As $\tilde{N}$ asymptotically vanishes over the cylindrical end of $\Omega_{0}$ and is zero over $S_{0}$ then, by the uniqueness of solutions to (\ref{LAP}), the combination has to be zero all over the set enclosed by $S_{0}$ and the cylindrical end. Then, the unique continuation principle \cite{MR0092067} tells that $\tilde{N}$ has to be zero all over $\Omega_{0}$ which is not possible because $f_{1}$ and $f_{2}$ were chosen to be linearly independent.  

The reason why we take such $N_{0}$ is twofold and will be explained adequately during the argumentation below.

In the space-time generated by the initial data consider the future-pointing congruence $\{\gamma(p,\tau)\}$ of time-like geodesics $\gamma(p,\tau)$ starting perpendicularly to $\Omega_{0}$ at $p\in \Omega_{0}$ and parametrized by proper time $\tau$. We are going to move $\Omega_{0}$ with the help of this congruence and obtain a foliation $\{\Omega_{t}\}$ \footnote{Of course is a foliation of a piece of the space-time.}. The leaves $\Omega_{t}$ of the foliation are defined, for every given $t$, as the image of the map
\ben
F_{t}:\, p\in \Omega_{0}\rightarrow \gamma(p,N_{0}(p)t)\in \Omega_{t}
\een  
This map in turn induces Lapse and Shifts, $N_{t},\, X_{t}$ over each $\Omega_{t}$ with the property that $N_{t=0}=N_{0}$ and $X_{0}=0$. Of course the result of moving a point $p\in\Omega_{0}$ through the space-time vector field $N_{\tau}{\bf n}_{\tau}+X_{\tau}$ and for a lapse of time $t$ is the same as $F_{t}(p)$. The leaves $\Omega_{t}$ are naturally identified to $\Omega_{0}$ and thus the space-time metric together with the electromagnetic tensor are described by a flow $(g_{t},K_{t};N_{t},X_{t};E_{t},B_{t})$ over $\Omega_{0}$ (c.f. Section \ref{BACKM} item {\bf (i)}; note also that we are changing notation from $(g(t),K(t);N(t),X(t);E(t),B(t))$ to $(g_{t},K_{t};N_{t},X_{t};E_{t},B_{t})$ which makes the writing clearer in this part). 

To simplify notation below, when we omit the subindex $t$ we mean $t=0$. 

We can comment now on one of the reasons why we chose $N_{0}$ satisfying (\ref{LAP}). In general, the time derivative of the mean curvature $k_{t}$ of the leaves of a space-like foliation $\{\Omega_{t}\}$ with Lapse $N_{t}$ and Shift $X_{t}$ is given by 
\ben
\partial_{t} k_{t}=-\Delta_{g_{t}} N_{t} + \big(4\pi ({\bf T}_{00}+{\bf T}_{ij}g_{t}^{ij}) + |K_{t}|^{2}\big)N_{t}
\een
In our case we have, at time $t$ equal zero, $\big(4\pi ({\bf T}_{00}+{\bf T}_{ij}g^{ij}) + |K|^{2}\big)=|E|^{2}$ (use ${\bf T}_{00}={\bf T}_{ij}g^{ij}$ and $8\pi {\bf T}_{00}=|E|^{2}+|B|^{2}$). Hence, $\partial_{t} k_{t}|_{t=0}=0$. As we also have $k_{t}|_{t=0}=0$ we obtain $k_{t}=(\partial_{t}^{2} k_{t}|_{t=0}) t^{2}/2 +O(t^{3})$ in short times. Having this quadratic behavior of $k_{t}$ in short times was one of the reasons behind the choice of $N_{0}$ and will be crucial later.

Define $S_{t}=F_{t}(S_{0})\subset \Omega_{t}$, the translation of $S_{0}$ by $F_{t}$. Recall tat we are identifying $\Omega_{t}$ to $\Omega_{0}$ through $F_{t}$. In this identification the surface $S_{t}$ is identified to $S_{0}$. In this sense the area of $S_{t}$ is the same as $A_{g_{t}}(S_{0})$, a notation that we keep using below. 

We claim that 
\be\label{SVTI}
\ddot{A}_{g_{t}}(S_{0})\bigg|_{t=0}=-A''_{N_{0}}(S_{0})
\ee
where the double dot means twice the $t$-derivative of $A_{g_{t}}(S_{0})$ and $A''_{N_{0}}(S_{0})$ is, following the notation introduced before, the second variation of area of $S_{0}$ along $N_{0}\zeta$. We prove this claim in what follows. As was calculated in Proposition 3 in \cite{ReirisZTBH} we have
\begin{align}\label{GRR}
\ddot{A}_{g_{t}}(S_{0})\bigg|_{t=0}=& \int_{S_{0}} \big[N_{0}\nabla_{A}\nabla_{B} N_{0} - N_{0}^{2}\big(Ric_{AB} -2K_{Ai}K^{i}_{\ B}\big)\big]\, h^{AB}\, dA\\ 
\nonumber & + \int_{S_{0}} 8\pi N_{0}^{2}\big[{\bf T}_{AB}-\frac{1}{2}( {\bf T}_{ij}g^{ij} - {\bf T}_{00}\big) g_{AB}\big]\, h^{AB}dA
\end{align}
where we included here the term involving ${\bf T}$ that was omitted in \cite{ReirisZTBH} as in there only vacuum solutions were considered \footnote{More precisely, in the second formula of Proposition 3 use $\dot{K}_{ij}=-\nabla_{i}\nabla_{j} N +N(Ric_{ij}-2K_{il}K^{l}_{\ j}) - 8\pi N({\bf T}_{ij}+\frac{1}{2}({\bf T}_{lm}g^{lm}-{\bf T}_{00})g_{ij})$ instead of just $\dot{K}_{ij}=-\nabla_{i}\nabla_{j} N +N(Ric_{ij}-2K_{il}K^{l}_{\ j})$ (recall that the data at the initial time is maximal, that is $k=0$). }. 
In the previous formula $Ric$ is the Ricci curvature of $g=g_{0}$ and $\nabla$ its covariant derivative. We note then that: 
\begin{enumerate}
\item The electromagnetic stress-energy is traceless and therefore ${\bf T}_{ij}g^{ij}-{\bf T}_{00}=0$,
\item $Ric_{AB}h^{AB}=R-Ric(\zeta,\zeta)=2|E|^{2}-Ric(\zeta,\zeta)$,
\item And finally, because $S_{0}$ has the geometry of an extreme RN-horizon the conditions (\ref{ERNE}) hold and we have 
\begin{align*}
& 8\pi {\bf T}_{AB}h^{AB}=2|E|^{2},\quad  K_{Ai}K^{i}_{\ B}h^{AB}=0,\quad \text{and,}\\
& \int_{S_{0}} N_{0}(\nabla_{A}\nabla_{B} N_{0})\, h^{AB}\, dA=-\int_{S_{0}} |\nabla N_{0}|^{2}\, dA
\end{align*}
where in the last formula the gradient of $N_{0}$ is taken over $S_{0}$.
\end{enumerate}
Combining this information in (\ref{GRR}) and after a crucial cancelation of the terms involving $|E|^{2}$ we obtain 
\ben
\ddot{A}_{g_{t}}(S_{0})\bigg|_{t=0}=-\int_{S_{0}} \bigg( |\nabla N_{0}|^{2} - Ric(\zeta,\zeta) N_{0}^{2}\bigg)\, dA=-A''_{N_{0}}(S_{0})
\een
where to deduce the second equality we have used (\ref{SIM}) and that $\Theta=0$ over $S_{0}$.
We can comment now on the second reason for our particular selection of $N_{0}$. If $N_{0}$ is not exactly the constant function one over $S_{0}$, as we are assuming, then $A''_{N_{0}}(S_{0})>0$ and therefore $\ddot{A}_{g_{t}}(S_{0})|_{t=0}<0$. This is our second reason and will be also crucial below.

The space-time vector field ${\bf V}$ which moves $\Omega_{0}$ to $\Omega_{t}$ and which generates the flow $g_{t}$, is,  at a space-time point $q=\gamma(p,N_{0}(p)t)$, given by 
\ben
{\bf V}(q)=\frac{d F(p,N_{0}(p)t)}{d t}=N_{0}(p) \frac{d\gamma(p,\tau)}{d \tau}\bigg|_{\tau=N_{0}(p)t}=N_{0}(p)\gamma'(q)
\een   
Recalling that $N_{0}$ tends to zero (indeed exponentially) over the asymptotically cylindrical end of $(\Omega_{0},g_{0})$ we conclude
that ${\bf V}$ tends to zero over the asymptotically cylindrical end and for this reason the evolution of $g_{t}$ over the end freezes up. Thus the metrics $g_{t}$ inherit exactly the same cylindrical asymptotic for every $t$, that is, that of the metric product of the unit two-sphere and the half-real line. 

Take (by continuity) $t^{*}>0$ small enough such that for all $t\in [0,t^{*}]$, the boundary of $(\Omega_{0},g_{t})$ is still strictly convex. Assume that $t^{*}$ was chosen small enough that $A_{g_{t}}(S_{0})<4\pi$ for every $t\in [0,t^{*}]$. Then, again based on general results on minimal surfaces \cite{MR678484} we can guarantee, for every $t\in [0,t^{*}]$, the existence of a stable minimal sphere 
\footnote{That the limit is connected and is a sphere follows from the genus bounds (1.4) of {\sc Theorem 1} in \cite{MR678484}.}
$\hat{S}_{t}$ in $\Omega_{0}$ of area less or equal than $A_{g_{t}}(S_{0})$, non contractible inside $\Omega_{0}$ and thus of electric charge one. 

We proceed now to gather conveniently all the information obtained so far and use it thereafter to reach a contradiction. 
\begin{enumerate}
\item From $k_{t}=(\partial^{2}_{t} k_{t}|_{t=0})t^{2}/2 +O(t^{3})$ we have, for all $t\in [0,t^{*}]$ (chose $t^{*}$ smaller if necessary), 
\be\label{UNO}
k^{2}_{t}\leq 2 c_{1}^{2} t^{4}\qquad \text{where}\qquad c_{1}=\sup \bigg\{\frac{\big| \partial^{2}_{t} k_{t}(p)\big|_{t=0}}{2}, p\in \Omega_{0}\bigg\} ,
\ee 
\item From $A_{g_{t}}(S_{0})=4\pi - A''_{N_{0}}(S_{0})t^{2}/2 + O(t^{3})$ we have, for all $t\in [0,t^{*}]$ (chose $t^{*}$ smaller if necessary), 
\be\label{DOS}
A_{g_{t}}(S_{0})\leq 4\pi - \frac{c_{2}}{2}t^{2}\leq 4\pi \qquad \text{where\ }\qquad c_{2}=\frac{A''_{N_{0}}(S_{0})}{2}>0
\ee
\item For every $t\in [0,t^{*}]$ there is a stable minimal sphere $\hat{S}_{t}$ with $Q_{\rm E}(\hat{S}_{t})=1$ and $A_{g_{t}}(\hat{S}_{t})\leq A_{g_{t}}(S_{0})$.

\end{enumerate}
Now, the stability inequality at $\hat{S}_{t}$ with trial function $\alpha=1$ gives
\ben
4\pi \geq \int_{\hat{S}_{t}} |E_{t}|^{2}\, dA_{t} - \int_{\hat{S}_{t}} \frac{k_{t}^{2}}{2}\, dA_{t}
\een  
Use then (\ref{UNO}) and that $\int_{\hat{S}_{t}}|E_{t}|^{2}\, dA_{t}\geq (4\pi)^{2}/A(\hat{S}_{t})$ (because $Q_{\rm E}(\hat{S}_{t})=1$) to transform this equation into
\ben
4\pi\geq \frac{(4\pi)^{2}}{A(\hat{S}_{t})} - c_{1}^{2}t^{4} A(\hat{S}_{t})
\een
Multiply this equation by $A(\hat{S}_{t})/4\pi$ and then use that $A(\hat{S}_{t})\leq A(S_{0})$ and (\ref{DOS}) to deduce 
$4\pi - c_{2}t^{2}/2\geq 4\pi -4\pi c^{1}_{2}t^{4}$ or, the same, $8\pi c^{2}_{2}t^{4}\geq c_{1}t^{2}$, which is impossible for small $t$. 
\end{proof}

\section{A family of perturbations of the ERN$_{1}$ initial data}\label{SECDS}

Recall that the metric of the ERN$_{1}$ space-time is 
\begin{align}\label{RNM}
{\bf g}=-\big(1-1/r\big)^{2}dt^{2}+\frac{1}{\big(1-1/r\big)^{2}}dr^{2}+r^{2}d\Omega^{2}
\end{align}
and that on the hypersurface $\Sigma_{0}=\{t=0\}$ we have $K_{0}=0,\ B_{0}=0$ and that the electric field is radial and takes the form $E_{0}=\zeta/r^{2}$ where $\zeta=\partial_{r}/|\partial_{r}|$ is the unit normal to the radial spheres $S_{\bar{r}}=\{r=\bar{r}\}$. 
Now, the constraint equations (\ref{CEEM}) of an electro-vacuum data set $(g,K;E,B)$ with $K=0$ and $B=0$ reduce to
\be\label{CEV}
\left\{
\begin{array}{l}
R=2|\, E\, |^{2},\vspace{.1cm} \\
{\rm div}\, E=0
\end{array}
\right.
\ee
Because of this the scalar curvature $R_{0}$ of the metric $g_{0}$ of the ERN$_{1}$ standard initial data is $R_{0}=2/r^{4}$. 

In the argumentation given below we will make use of an expression for the three-Laplacian $\Delta_{g_{0}}$ acting on radial functions $\phi=\phi(r)$ of $\Sigma_{0}$. A direct calculation using the general formula $\Delta \phi=\partial_{r}(\sqrt{g}g^{rr}\partial_{r} \phi)/\sqrt{g}$ gives, when $\phi=\phi(r)$, the expression
\ben
\Delta_{g_{0}}\, \phi=\frac{r(r-1)}{r^{4}}\frac{d}{dr}\bigg[r(r-1)\frac{d}{dr}\phi\bigg]
\een
This formula is simplified if we use the harmonic radial coordinate $x=\ln(1-1/r)$ instead of $r$ (harmonic means $\Delta_{g_{0}} x=0$). With this definition the range of $x$ is $(-\infty,0)$. In this new coordinate the Laplacian acting on radial functions reads 
\be\label{LSF}
\Delta_{g_{0}}\, \phi=\frac{\phi''}{r^{4}}
\ee
where here $\phi''=d^{2}\phi/dx^{2}$. Note then that $\Delta_{g_{0}}\, \phi=|E_{0}|^{2}\,\phi''$.

We proceed now to construct the bi-parametric family of axisymmetric ``perturbations'' of the initial data on $\Sigma_{0}$. The axisymmetric Killing field will be $\partial_{\varphi}$, which, note, is also axisymmetric Killing for the background data set. 
The two parameters of the family will be $\hat{\epsilon}$ and $\hat{x}$. Roughly speaking the variable $\hat{\epsilon}$ represents the ``strength'' of the perturbation while $\hat{x}$ marks the sphere around which the perturbation ``concentrates''. This interpretation will be clear as the construction progresses. To explain the construction let us recall in what follows the {\it conformal method} to solve the constraint equations but for the situation that is of interest here, namely when the data set to be found is time symmetric and has no magnetic field. Let $(\Sigma,g)$ be a Riemannian manifold of scalar curvature $R$. On it let $\hat{E}$ be a $g$-divergence-less vector field. If for $\phi>0$ we have
\be\label{LEQ}
\Delta\, \phi = R\,\phi - 2|\,\hat{E}\,|^{2}\, \phi^{-3},
\ee
then $\bar{g}=\phi^{4} g$ and $\bar{E}=\phi^{-6}\hat{E}$ satisfy the constrain equations (\ref{CEV}). We will use this method below
with $(\Sigma,g)=(\Sigma_{0},g_{0})$ and $\hat{E}=E_{\hat{\epsilon},\hat{x}}$ suitably chosen.

In what follows we will identify $\Sigma_{0}$ to $(-\infty,0]\times \mathbb{S}^{2}$ where the factor $(-\infty,0]$ is the range of the coordinated $x$ introduced before. 
From now on the parameter $\hat{x}$ is set to vary in $(-\infty,-2]$ and $\hat{\epsilon}$ in $(0,1/16)$. Fix a smooth and non-zero axisymmetric two-form $\omega$ supported on $(-3,-1)\times {\mathbb S}^{2}\subset (-\infty,-1)\times {\mathbb S}^{2}$. This form is set to be fixed from now on and will not be adjusted anymore. 
For every $\hat{x}$ let $\chi_{\hat{x}}^{*}\, \omega$ be the pull-back of $\omega$ to $[\hat{x}-1,\hat{x}+1]\times {\mathbb S}^{2}$ under the transformation $\chi_{\hat{x}}:[\hat{x}-1,\hat{x}+1]\times S^{2}\rightarrow [-3,-1]\times S^{2}$ given by $(x,\theta,\varphi)\rightarrow (x-\hat{x}-2,\theta,\varphi)$.
Then, for every $\hat{x}$ and $\hat{\epsilon}$ define
\be\label{EHAT}
\hat{E}_{\hat{x},\hat{\epsilon}}=E_{0}+\hat{\lambda}\, \big(\hspace{-.1cm}\star d \star (\chi^{*}_{\hat{x}}\, \omega)\big)^{\sharp}
\ee
where $\hat{\lambda}=\hat{\lambda}_{\hat{x},\hat{\epsilon}}$ is a factor chosen to have $\hat{\epsilon}=\sup\, \big|1-|\hat{E}_{\hat{x},\hat{\epsilon}}|^{2}/|E_{0}|^{2}\big|$ (here $|\ldots|=|\ldots|_{g_{0}}$), the star $\star$ in $\star d\star $ is the $g_{0}$-Hodge star and $(\star d \star (\chi^{*}_{\hat{x}}\, \omega))^{\sharp}$ is the $g_{0}$-dual vector field of the form $\star d \star (\chi^{*}_{\hat{x}}\, \omega)$. In this way $\hat{E}_{\hat{x},\hat{\epsilon}}$ comprises a bi parametric family of divergence-less axisymmetric vector fields which are equal to the background field $E_{0}$ outside $[\hat{x}-1,\hat{x}+1]\times {\mathbb S}^{2}$ but otherwise not very different from it. 

In what follows and to simplify notation we keep using $|\ldots|=|\ldots|_{g_{0}}$ and make also $\hat{E}=\hat{E}_{\hat{x},\hat{\epsilon}}$.

We pass now to show that for every $\hat{E}$ we can find an axisymmetric solution to the the Lichnerowitz equation (\ref{LEQ}) (L-equation from now on) with good geometric properties. To this extent we use the method of sub and super-solutions. Namely, if for axisymmetric functions (barriers) $\phi_{+}>0$ and $\phi_{-}>0$ with $\phi_{+}>\phi_{-}$ we have 
\begin{align}\label{SSSOL}
\left\{
\begin{array}{l}
\Delta_{g_{0}} \phi_{+}\leq 2|E_{0}|^{2}\phi_{+}-2|\hat{E}|^{2}\phi_{+}^{-3},\vspace{.3cm}\\ 
\Delta_{g_{0}} \phi_{-}\geq 2|E_{0}|^{2}\phi_{-}-2|\hat{E}|^{2}\phi_{-}^{-3}
\end{array}
\right.
\end{align}
(recall $R_{0}=2|E_{0}|^{2}$) then there is an axisymmetric solution $\phi>0$ to (\ref{LEQ}) with $\phi_{-}\leq \phi\leq \phi_{+}$, (for a proof of this fact in this context see \cite{Chrusciel-Mazzeo}\footnote{To get an axisymmetric solution out of the method of barriers just work inside the family of axisymmetric functions all the time in \cite{Chrusciel-Mazzeo}.}). We explain now how to find $\phi_{-}$ which will be a radial function, i.e. $\phi_{-}=\phi_{-}(x)$. In (I) below we define $\phi_{-}(x)$ over $(-\infty,-1]$ and in (II) over $[-1,0)$. 
The global function defined by (I) and (II) will be smooth over the separate domains $(-\infty,-1)$ and $(-1,0)$ but will be just $C^{0}$ at $x=-1$. For this reason to check that such global function is a barrier in the distributional sense \cite{Chrusciel-Mazzeo} it will be necessary to check that its left derivative at $x=-1$ is less than its right derivative \footnote{Alternatively, a smooth barrier can be easily found by rounding off the global function constructed by (I) and (II).}. This will be done after (I) and (II) below.
\begin{figure}[h]
\centering
\includegraphics[width=10cm,height=2.5cm]{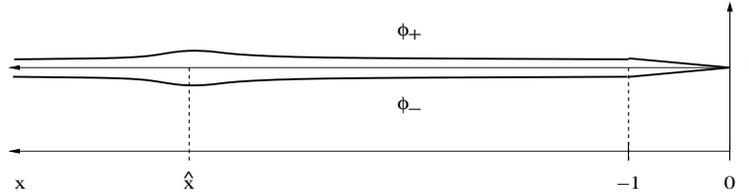}
\caption{Picture of the barriers $\phi_{-}$ and $\phi_{+}$.}
\label{Figure3}
\end{figure}
\begin{enumerate}
\item[(I)] {\it Defining $\phi_{-}(x)$ on $(-\infty,-1]$}. Make $\psi_{-}=\phi_{-}-1$ and recall that $\Delta \psi_{-}=|E_{0}|^{2} \psi_{-}''$. 
With this information and after a simple manipulation the second equation in (\ref{SSSOL}) can be displayed in the form     
\be\label{L2}
\psi_{-}''\geq 2\, \big[1+\phi_{-}^{-1}+\phi_{-}^{-2}+\phi_{-}^{-3}\big]\, \psi_{-} + 2\, \bigg[\, 1-\frac{|\hat{E}|^{2}}{|E_{0}|^{2}}\, \bigg]\, \phi_{-}^{-3}
\ee
Now, it can be easily checked that for any real number $\gamma$ such that $|\gamma-1|\leq 1/8$ we have
\be\label{EI}
3\leq 1+\gamma^{-1}+\gamma^{-2}+\gamma^{-3}\leq 5,\quad \text{and}\quad \frac{1}{2}\leq \gamma^{-3}\leq 2
\ee
Thus, if we can find $\psi_{-}(x)$ with $-1/8<\psi_{-}<0$ and satisfying 
\be\label{LB1}
\psi_{-}''\geq 6\psi_{-}+4\hat{\epsilon} \hat{I}
\ee
where $\hat{I}=\hat{I}(x)$ is the {\it indicator function} on $[\hat{x}-1,\hat{x}+1]$, (i.e. equal to one on $[\hat{x}-1,\hat{x}+1]$ and zero otherwise), then $\phi_{-}=1+\psi_{-}$ will verify (\ref{L2}) because, in this case, we would have
\ben
6\psi_{-}+4\hat{\epsilon} \hat{I}\geq 2\, \big[ 1+\phi_{-}^{-1}+\phi_{-}^{-2}+\phi_{-}^{-3}\big]\, \psi_{-} + 2\, \bigg[\, 1-\frac{|\hat{E}|^{2}}{|E_{0}|^{2}}\, \bigg]\, \phi_{-}^{-3}
\een
due to (\ref{EI}) (with $\gamma=\phi_{-}$) and because, by construction, we have $1-|\hat{E}|^{2}/|E_{0}|^{2}\leq \hat{\epsilon}$ point-wise. The function
\ben
\psi_{-}(x)=\frac{-4\hat{\epsilon}}{\cosh\, (x-\hat{x})}
\een
verifies $-1/8<\psi_{-}(x)<0$ because $\hat{\epsilon}< 1/16$. To see that it also satisfies (\ref{LB1}) on $(-\infty,-1]$ we argue as follows. First we compute $\psi''_{-}=4\hat{\epsilon} (1-2\sinh^{2} (x-\hat{x})/\cosh^{2} (x-\hat{x}))/\cosh (x-\hat{x})$ and, after plugging this inside (\ref{LB1}) and after a simple manipulation we conclude that to verify (\ref{LB1}) it is enough to verify the inequality $7-2\sinh^{2}(x-\hat{x})/\cosh^{2}(x-\hat{x})\geq \big(\cosh (x-\hat{x})\big)\, \hat{I}(x)$ for all $x\in (-\infty,-1]$. This is easily seen because the l.h.s of this expression is greater than five and the r.h.s is less or equal than $\cosh 1$ which is less than $e$. Summarizing, $\phi_{-}=\psi_{-}+1$ is a sub-solution on this range of $x$. Note that as $\psi_{-}<0$ then it is $\phi_{-}<1$.  

\vs
\item[(II)] {\it Defining $\phi_{-}(x)$ on $[-1,0)$}. On $[-1,0)$ define $\phi_{-}(x)$ by $\phi_{-}(x)=1+\psi_{-}(x)$ where
\ben
\psi_{-}(x)=\frac{4\hat{\epsilon} x}{\cosh\, (-1-\hat{x})}
\een
To see that $\phi_{-}$ is a sub-solution it is necessary to check (\ref{L2}). Firstly, as $\phi_{-}(x)$ is linear in $x$ the l.h.s of (\ref{L2}) is zero. Secondly, the second term on the r.h.s of (\ref{L2}) is zero because when $x\in [-1,0)$ it is $|E_{0}|^{2}=|\hat{E}|^{2}$. The inequality (\ref{L2}) then follows because $\psi_{-}<0$ and so is the first term on the r.h.s of (\ref{L2}).  
\end{enumerate}

So far we have defined $\phi_{-}$ and proved that it is a sub-solution when restricted to the intervals $(-\infty,-1)$ and $(-1,0)$. It remains to prove that it is also a sub-solution in the neighborhood of $x=-1$. As said, to see this it is enough to check that the left-sided derivative of $\phi_{-}$ at $x=-1$ is less than its right-sided derivative.  The left-sided derivative at $x=-1$ is $4\hat{\epsilon}\sinh (-1-\hat{x})/\cosh^{2}(-1-\hat{x})$ while the right-sided is $4\hat{\epsilon}/\cosh(-1-\hat{x})$ and the desired inequality follows.

Summarizing, the sub-solution is
\be\label{LB}
\phi_{-}(x)=\left\{
\begin{array}{lll}
1-{\displaystyle \frac{4\hat{\epsilon}}{\cosh (x-\hat{x})}} & \text{ if } & x\in (-\infty,-1],\vspace{0.2cm} \\
1+{\displaystyle \frac{4\hat{\epsilon} x}{\cosh (-1-\hat{x})}} & \text{ if } & x\in [-1,0)
\end{array}
\right.
\ee
A graph of $\phi_{-}$ is presented in Figure \ref{Figure3}. Reproducing the argument that lead to $\phi_{-}$, it is found that $\phi_{+}(x)$, defined by 
\be\label{UB}
\phi_{+}(x)=\left\{
\begin{array}{lll}
1+{\displaystyle \frac{4\hat{\epsilon}}{\cosh (x-\hat{x})}} & \text{ if } & x\in (-\infty,-1],\vspace{0.2cm} \\
1-{\displaystyle \frac{4\hat{\epsilon} x}{\cosh (-1-\hat{x})}} & \text{ if } & x\in [-1,0)
\end{array}
\right.
\ee
is a super-solution. We conclude that there is $\phi>0$, solution of (\ref{LEQ}), and satisfying $\phi_{-}\leq \phi \leq \phi_{+}$. The metric $\bar{g}=\phi^{2}g_{0}$ and the electric field $\bar{E}=\hat{E}\phi^{-6}$ satisfy the constraint equations (\ref{CEV}). 

Summarizing, from the explicit form of the sub and super-solutions we observe that $\phi-1$ ``concentrates`` around $\hat{x}$ and decays exponentially to zero in both directions of $x$ starting from $\hat{x}$. In the direction of increasing $x$ the exponential decay however stops at $x=-1$ and after that it is linear in $x$, namely of the order $1/r$ in the $r$-coordinate. Observe, to be recalled later, that the exponential 
decay of $\phi$ in the asymptotically cylindrical end implies by standard elliptic estimates that the perturbed data sets $(\bar{g},\bar{K};\bar{E},\bar{B})$ decay exponentially as defined in the introduction.

\section{Rigidity of the ERNT$_{1}$ initial data}\label{RIGS} 

The next lemma shows the rigidity of the ERNT space-time and has interest in itself. It will be used in the proof of Proposition \ref{P2}.

\begin{Lemma}\label{PP} Let $(\Sigma;g,K;E,B)$ be a smooth complete and maximal electro-vacuum data set where $\Sigma$ is diffeomorphic to $\mathbb{R}\times {\mathbb S}^{2}$. Let $S_{0}:=0\times {\mathbb S}^{2}$ and suppose that $|Q_{\rm E}(S_{0})|=1$. Suppose too that for any (compact, boundaryless and embedded) surface $S$ non-contractible inside $\Sigma$ we have $A(S)\geq 4\pi$, and that there is at least one such $S$ with $A(S)=4\pi$. Then the data set is the standard ERNT$_{1}$ initial data.
\end{Lemma} 
\n For expository reasons it is better to divide the proof into three Auxiliary Propositions. In every one of them we let ${\mathcal F}$ {\it be the set of (compact, boundaryless and embedded) surfaces of area $4\pi$ and which are non-contractible inside $\Sigma$}.

\begin{AProposition}\label{AUXP1} Assume the hypothesis of Lemma \ref{PP}. 
Then, each $S\in {\mathcal F}$ is a (normalized) ERN sphere and every two different spheres in ${\mathcal F}$ are disjoint. Moreover the set $\bigcup_{S\in {\mathcal F}} \{S\}$ is closed as a set in $\Sigma$. 
\end{AProposition}

\begin{AProposition}\label{AUXP2} Assume the hypothesis of Lemma \ref{PP}.
If $\bigcup_{S\in {\mathcal F}} \{S\}= \Sigma$ then the data set is the standard ERNT$_{1}$ initial data $(\check{\Sigma};\check{g}_{0},\check{K}_{0};\check{E}_{0},\check{B}_{0})$.
\end{AProposition}

\begin{AProposition}\label{AUXP3} Assume the hypothesis of Lemma \ref{PP}. Then, $\bigcup_{S\in {\mathcal F}} \{S\}= \Sigma$.
\end{AProposition}
The proofs of the three propositions are presented consecutively.

\begin{proof}[\bf Proof of Aux-Proposition \ref{AUXP1}.]  By the hypothesis of Lemma \ref{PP} every non-contractible surface has area greater or equal than $4\pi$. Therefore the surfaces in ${\mathcal F}$, which have area equal to $4\pi$, must be minimal and stable (see footnote \ref{FOTN1}). By Lemma \ref{RNHG} they are (normalized) ERN spheres. We show next that two different spheres $S_{1}$ and $S_{2}$ in ${\mathcal F}$ (in case ${\mathcal F}$ has more than one element) must be disjoint. If $S_{1}\cap S_{2}\neq \emptyset$ then, being minimal surfaces, they must intersect transversely. We will think the surfaces $S_{i}$, $i=1,2$ as embedded in $\big(\mathbb{R}^{3}\setminus \{o\}\big)\sim \Sigma$. As the $S_{i}, i=1,2$ are non-contractible inside $\mathbb{R}^{3}\setminus \{o\}$ then there are open balls $B_{1}$ and $B_{2}$ in $\mathbb{R}^{3}$ containing the origin $o$ and such that $\partial \overline{B_{i}}=S_{i}$ for $i=1,2$. Define the manifolds
\begin{align*}
& {\mathcal V}_{1}:=S_{1}\cap {\rm Int}(B_{2}^{c}),\quad {\mathcal V}_{2}:=S_{2}\cap {\rm Int}(B_{1}^{c})\quad \text{and}\quad {\mathcal W}_{1}:=S_{1}\cap B_{2},\quad {\mathcal W}_{2}:=S_{2}\cap B_{1}
\end{align*}
where ${\rm Int}(B_{i}^{c})$ is the interior of the complement of $B_{i}$ (see Figure \ref{Figure4}). The manifolds ${\mathcal V}_{1}, {\mathcal V}_{2}, {\mathcal W}_{1}$ and ${\mathcal W}_{2}$ are pairwise disjoint and their closures have the same boundary. We will denote such boundary (a union of embedded circles indeed) by ${\mathcal B}$. 
We have $S_{1}=\overline{\mathcal V}_{1}\cup \overline{\mathcal W}_{1}$ and $S_{2}=\overline{\mathcal V}_{2}\cup \overline{\mathcal W}_{2}$ and for this reason it is
\be\label{PIP}
4\pi= A(\overline{\mathcal V}_{1})+A(\overline{\mathcal W}_{1}) \quad \text{and}\quad 4\pi=A(\overline{\mathcal V}_{2})+A(\overline{\mathcal W}_{2}).
\ee
The manifolds 
\ben
{\mathcal V}:=\overline{\mathcal V}_{1}\cup \overline{\mathcal V}_{2}\quad \text{and}\quad
{\mathcal W}:=\overline{\mathcal W}_{1}\cup \overline{\mathcal{W}}_{2},
\een
are embedded and smooth except at ${\mathcal B}$, where they have necessarily corners. Note that ${\mathcal V}$ and ${\mathcal W}$ are not necessarily connected (see Figure \ref{Figure4}). Moreover we have ${\mathcal V}=\partial (\overline{B_{1}\cup B_{2}})$ and ${\mathcal W}=\partial (\overline{B_{1}\cap B_{2}})$. Therefore, as $o\in B_{1}\cup B_{2}$ and $o\in B_{1}\cap B_{2}$, then at least one of the connected component of ${\mathcal V}$ and at least one of ${\mathcal W}$ divide $\mathbb{R}^{3}\setminus \{o\}$ into two connected components and are consequently non-contractible inside $\Sigma$.  
By (\ref{PIP}) if $A({\mathcal V}_{2})\leq A({\mathcal W}_{1})$ then $A({\mathcal V})\leq 4\pi$, while if $A({\mathcal V}_{2})\geq A({\mathcal W}_{1})$ then $A({\mathcal W})\leq 4\pi$. In any case we can round off the corners at ${\mathcal B}$ of either the manifold ${\mathcal V}$ or the manifold ${\mathcal W}$ to obtain one of area less than $4\pi$ and having at least one connected component non-contractible inside $\Sigma$. This is against hypothesis and therefore the surfaces $S_{1}$ and $S_{2}$ have to be disjoint. 
\begin{figure}[h]
\centering
\includegraphics[width=4cm,height=5.5cm]{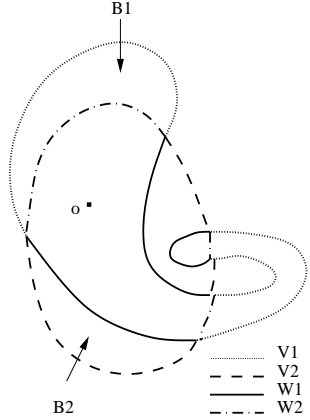}
\caption{Representation of the manifolds ${\mathcal V}_{1}$, ${\mathcal V}_{2}$, ${\mathcal W}_{1}$ and ${\mathcal W}_{2}$.}
\label{Figure4}
\end{figure}

It remains to be proved that the set $\bigcup_{S\in {\mathcal F}} \{S\}$ is closed in $\Sigma$. But if $p_{i}(\in S_{i}\in {\mathcal F})$ is a sequence of points in $\bigcup_{S\in {\mathcal F}} \{S\}$ with limit point $p_{\infty}$, then the sequence $S_{i}$ of (normalized) 
ERN spheres, and therefore of stable and area minimizing minimal surfaces, has a subsequence converging (in $C^{k}$ for every $k\geq 1$)  to a limit stable minimal sphere $S_{\infty}\ni p_{\infty}$, \cite{MR1068795}\footnote{Precisely there are embeddings $f_{i}:\mathbb{S}^{2}\rightarrow S_{i}$ converging  in $C^{k}$ to a covering immersion $f_{\infty}:\mathbb{S}^{2}\rightarrow S_{\infty}$. But in our case $\Sigma\sim \mathbb{R}^{3}\setminus \{o\}$ and therefore $S_{\infty}$ must be orientable, hence a sphere and $f_{\infty}$ an embedding.}. The sphere $S_{\infty}$ cannot be contractible inside $\Sigma$ otherwise the $S_{i}$'s would be contractible for  sufficiently big $i$. We have $4\pi=\lim A(S_{i})=A(S_{\infty})$, thus $S_{\infty}\in {\mathcal F}$ and therefore $p_{\infty}\in \bigcup_{S\in {\mathcal F}}\{S\}$.  
\end{proof}

\begin{proof}[\bf Proof of Aux-Proposition \ref{AUXP2}.]  Assume at the moment that the foliation ${\mathcal F}$ is smooth (see the definition of smooth foliation in \cite{MR824240}). We will be proving this later. Fix a sphere $S^{*}$ in ${\mathcal F}$ and denote by $\Sigma^{*}_{L}$ and $\Sigma^{*}_{R}$ the connected components of $\Sigma\setminus S^{*}$. For any $p\in \Sigma$ let $S(p)$ be the sphere in ${\mathcal F}$ containing $p$ and denote by $\Omega(p)$ the region enclosed by $S^{*}$ and $S(p)$. Then define the (smooth) function $\tilde{x}:\Sigma\rightarrow \mathbb{R}$ as 
\ben
\tilde{x}(p)=
\left\{
\begin{array}{lcl}
\ \ \, {\rm Vol}(\Omega(p)) & \text{ if } & p\in \Sigma^{*}_{R},\vs\\
-{\rm Vol}(\Omega(p)) & \text{ if } & p\in \Sigma^{*}_{L}
\end{array}
\right.
\een
This function is constant over every leaf and has nowhere zero gradient \footnote{This can be easily seen from the the fact that ${\mathcal F}$ is assumed smooth.}. Let $Y=\nabla \tilde{x}/|\nabla \tilde{x}|^{2}$ and note that as $Y(\tilde{x})=1$ the flow induced by $Y$ carries leaves (of ${\mathcal F}$)  into leaves (of ${\mathcal F}$). 
Fix an isometry $\psi:{\mathbb S}^{2}\rightarrow S^{*}$ and define the diffeomorphism $\Phi:\mathbb{R}\times \mathbb{S}^{2}\rightarrow \Sigma$ by sending a pair $(t,s)$ into the translation of $\psi(s)$ through the flow induced by $Y$ and by a parametric time $t$. 
Of course we have $\Phi_{*}\partial t=Y$. 
On the other hand if we denote by $h_{\tilde{x}}$ the induced metric on the leaves, then we have 
${\mathcal L}_{Y} h_{\tilde{x}}=0$ because each leaf is totally geodesic (here ${\mathcal L}$ is the Lie-derivative). Therefore we can write 
\ben
\Phi^{*}g=|\nabla \tilde{x}|^{2} d\tilde{x}^{2}+d\Omega^{2}
\een
We show now that $|\nabla \tilde{x}|$ is constant over every leaf. Indeed, as the areas of the spheres of ${\mathcal F}$ is $4\pi$ then the second variation of area of any sphere in ${\mathcal F}$ along $Y$ is zero, i.e. $A''_{Y}(S)=0$. This implies that $|Y|=1/|\nabla \bar{x}|$ is constant over every sphere (see the proof of Lemma \ref{RNHG}). The metric (\ref{RNMT}) is recovered by making a simple change of variables $\bar{x}=\bar{x}(\tilde{x})$, with $|\nabla \tilde{x}|=d\bar{x}/d\tilde{x}$. Finally by Lemma \ref{RNHG} we have $B=0$, $K=0$ and $E=\zeta$ with $\zeta$ a normal field to the leaves of ${\mathcal F}$ (i.e. either $\partial_{\bar{x}}$ or $-\partial_{\bar{x}}$). Hence we have $(g,K;E,B)=(\check{g},\check{K};\check{E},\check{B})$ as claimed. 

It remains to prove that the foliation ${\mathcal F}$ is smooth. We will show that the 1-distribution of lines perpendicular to the leaves of ${\mathcal F}$ is smooth. This implies that the distribution of the tangent planes to the leaves of ${\mathcal F}$ is smooth and the smoothness of ${\mathcal F}$ is then direct from Frobenius's theorem \cite{MR824240}. Let $S$ be a sphere of ${\mathcal F}$, let $\zeta$ be a normal field to it and let $h$ be the induced two-metric.  We will show that the Ricci curvature $Ric$ of $g$ over $S$ has the following form: $Ric(\zeta,\zeta)=0$ and for any $v,w\in TS$ we have $Ric(\zeta,v)=0$ and $Ric(v,w)=h(v,w)$. The 1-distribution of normal directions to ${\mathcal F}$ is then uniquely characterized by the null space of $Ric$ (i.e. $\{v\in TS, Ric(v,v)=0\}$), and is easily seen to be smooth because $Ric$ is smooth. 

Again let $S$ be a surface in ${\mathcal F}$ and $\zeta$ a unit normal field to it. Let $\{\gamma_{q}(\tau),q\in S, 0\leq \tau\leq \tau_{0}\}$ be the congruence of geodesics in $\Sigma$ starting at $\tau=0$ perpendicularly to $S$ in the direction of $\zeta$ and parametrized by the arc-length $\tau$. We will move $S$ by the vector field $V=\partial_{\tau} \gamma_{q}(\tau)$ and obtain a smooth one-parametric family of surfaces $S(\tau)$. We assume that $\tau_{0}$ is small enough that the surfaces $S(\tau)$ are embedded (and smooth).

In the forthcoming equations, but inside this proof, we will denote the mean curvature ${\rm tr}_{h}\Theta$ by $\mu$. Recall from Lemma \ref{RNHG} that over $S$ we have $\kappa=1$, $R=2$ and $\Theta=0$. Therefore from the general identity 
\be\label{GCO}
2\kappa - |\Theta|^{2}+\mu^{2}=R-2Ric(\zeta,\zeta)
\ee
we obtain $Ric(\zeta,\zeta)=0$. Also from ${\rm div}\, \Theta - d\mu =Ric(\zeta, -)$ we obtain $Ric(\zeta,v)=0$ for any $v\in TS$. To show that for any $v,w\in TS$ we have $Ric(v,w)=h(v,w)$ it is enough to prove that $\mathcal{L}_{V}{\Theta}=\dot{\Theta}=0$ because of the general identity (on $TS$)
\ben
\dot{\Theta}=-\mu \Theta + 2\Theta\circ \Theta + \kappa h - Ric
\een 
which gives $\dot{\Theta}(0)=h-Ric$ at $\tau=0$. Now, at any time $\tau\in (0,\tau_{0})$ we have
\begin{align}
\ddot{A}(S(\tau))= & \int_{S(\tau)} (\dot{\mu}+\mu^{2})\, dA\\
\nonumber = & \int_{S} \bigg(-\frac{|\Theta|^{2}}{2} +\frac{\mu^{2}}{2} + \kappa - |E|^{2} \bigg)\, dA \\ 
\nonumber = & \int_{S(\tau)} \bigg(-\frac{|\Theta|^{2}}{2} +\frac{\mu^{2}}{2} \bigg)\, dA + \bigg[4\pi - \int_{S(\tau)} |E|^{2}\, dA \bigg] \\ 
\nonumber \leq & \underbrace{\int_{S(\tau)} \bigg(-\frac{|\Theta|^{2}}{2} +\frac{\mu^{2}}{2} \bigg)\, dA}_{\displaystyle U(\tau)}
\end{align}
where: (i) to obtain the first inequality we use $d\dot{A}=\mu\, dA$, (ii) to pass from the second to the third line we use the focussing  (Riccati) equation $\dot{\mu}=-|\Theta|^{2}-Ric(\zeta,\zeta)$ in conjunction with (\ref{GCO}) and $R\geq 2|E|^{2}$, (iii) to pass from the second to the third line we use Gauss-Bonnet and (iv) from the third to the fourth we use (\ref{EESI}). On the other hand we can express $A(S(\tau))$ as
\ben
A(S(\tau))=4\pi +\int_{0}^{\tau} d\tilde{\tau}\int_{0}^{\tilde{\tau}} \ddot{A}(S(\tilde{\tilde{\tau}}))\, d\tilde{\tilde{\tau}}
\een
and we have $\ddot{A}(S(\tau))\leq U(\tau)=U(0)+U'(0)\tau+U''(0)\tau^{2}/2+O(\tau^{3})$ with $U(0)=0$, $U'(0)=0$ and 
\ben
U''(0)=-\frac{1}{2}\int_{S(0)} |\dot{\Theta}(0)|^{2}\, dA
\een
as can be easily seen using $\mu(0)=0$, $\dot{\mu}(0)=0$, $\Theta(0)=0$ and $\dot{\Theta}(0)=h-Ric$. Therefore, if $\dot{\Theta}(0)\neq 0$ then we would have $A(S(\tau))<4\pi$ for small $\tau$ which is against the hypothesis. This finishes the proof. 
\end{proof}

\begin{proof}[\bf Proof of Aux-Proposition \ref{AUXP3}.] 
We will proceed by contradiction and assume that $\bigcup_{S\in {\mathcal F}} \{S\}\neq \Sigma$. As by Aux-Proposition \ref{AUXP1} the set $\bigcup_{S\in {\mathcal F}} \{S\}$ is closed, then every connected component of $\Sigma\setminus \bigcup_{S\in {\mathcal F}}\{S\}$ is either an open region enclosed by two spheres in ${\mathcal F}$ or an open region enclosed by a sphere in ${\mathcal F}$ and one of the two ends of $\Sigma$. Thus, if there is only one connected component of $\Sigma\setminus \bigcup_{S\in {\mathcal F}}\{S\}$ then $\bigcup_{S\in {\mathcal F}}\{S\}$ must at least contain a closed region enclosed by a sphere in ${\mathcal F}$ and one end of $\Sigma$. As in Aux-Proposition \ref{AUXP2} the data set over such region must be $ERNT_{1}$. Because of this one can cut off such region and ``double'' the remaining one to construct a new data set $(\Sigma';g',K';E',B')$ in the hypothesis of Lemma \ref{PP} but with two connected components of $\Sigma'\setminus \bigcup_{S\in {\mathcal F}}\{S\}$.  

Assume then without loss of generality that there are at least two connected components of $\Sigma\setminus \bigcup_{S\in {\mathcal F}} \{S\}$. We want to prove that such data set cannot exist. This will be done exactly as in Proposition \ref{PAS}. For this reason the paragraphs below are first dedicated to construct a setup similar to the one in the proof of Proposition \ref{PAS}.   

For the discussion that follows the Figure \ref{Figure5} could be of great help. Denote two of the connected components of $\bigcup_{S\in {\mathcal F}}\{S\}$ by $\Omega_{L}$ and $\Omega_{R}$ ($L$ for ``Left'' and $R$ for ``Right''). Let $S_{L}$ and $S_{R}$ be any two  spheres embedded in $\Omega_{L}$ and $\Omega_{R}$ respectively and non contractible inside $\Sigma$. Denote by $\Omega_{LR}$ the region enclosed by them and including them, and by $\Sigma^{-}_{L}$ (resp. $\Sigma^{+}_{R}$) the connected component of $\Sigma\setminus S_{L}$ (resp. $\Sigma\setminus S_{R}$) not containing $S_{R}$ (resp. $S_{L}$). Also let $D^{*}>0$ be small enough such that
\begin{enumerate}
\item if $p\in \Sigma_{L}^{-}$ (resp. $p\in \Sigma_{R}^{+}$) and ${\rm dist}(p,S_{L})\leq D^{*}$ (resp. ${\rm dist}(p,S_{R})\leq D^{*}$) then $p\in \Omega_{L}$ (resp. $p\in \Omega_{R}$), and
\item for any $0<D\leq D^{*}$ the set $\{p\in \Sigma_{L}^{-},{\rm dist}(p,S_{L})=D\}$ (resp. $\{p\in \Sigma_{R}^{+},{\rm dist}(p,S_{R})=D\}$) is a smooth and embedded sphere. 
\end{enumerate}     
In this context define the sphere $S^{*}_{L}$ (resp. $S^{*}_{R}$) as $S^{*}_{L}:=\{p\in \Sigma_{L}^{-},{\rm dist}(p,S_{L})=D^{*}\}$ (resp. $S^{*}_{R}:=\{p\in \Sigma_{R}^{+},{\rm dist}(p,S_{R})=D^{*}\}$) and let $\Omega^{*}_{LR}$ be the set enclosed by $S^{*}_{L}$ and $S^{*}_{R}$ including them. As the components $\Omega_{L}$ and $\Omega_{R}$ are different there is at least one sphere $S_{0}\in {\mathcal F}$ embedded in $\Omega_{LR}$ and therefore in $\Omega^{*}_{LR}$. Now, on $\Omega_{LR}^{*}$ consider a positive solution $N=N_{0}$ of the maximal lapse equation 
\ben
\Delta N - \big(4\pi ({\bf T}_{00}+{\bf T}_{ij}g^{ij}) + |K|^{2}\big)N=\Delta N - (|E|^{2}+|B|^{2}+|K|^{2})N = 0 
\een
and that is not identically to a constant over $S_{0}$. The existence of such $N_{0}$ is shown in the same way as was done in Proposition \ref{PAS} and is left to the reader. Also in the same way as in Proposition \ref{PAS} construct from $N_{0}$ a time-like vector field ${\bf V}$ and from it a flow $(g_{t},K_{t};E_{t},B_{t})$ over $\Omega_{LR}^{*}$, with $0\leq t\leq t^{*}$ and for some $t^{*}$ small. As in Proposition \ref{PAS} now we have $\dot{A}_{g_{t}}(S_{0})=0$ and $\ddot{A}_{g_{t}}(S_{0})<0$. Therefore $A_{g_{t}}(S_{0})<4\pi$ in short times $t$.

Instead of $g_{t}$ we are going to consider a modified flow of metrics $\tilde{g}_{t}$ conformally related to $g_{t}$. This will help to guarantee the existence of certain stable minimal spheres. To the purpose of defining $\tilde{g}_{t}$ consider the following function of $z\in [0,D^{*}]$,
\ben
\Psi_{\delta}(z)=1+e^{\displaystyle -1/z + 1/(D^{*}+\delta-z)}
\een
where $\delta$ is a constant to be fixed soon below. Observe that $\Psi_{\delta}(0)=1$ and that all the right-sided derivatives of $\Psi_{\delta}$ are zero at $z=0$. Observe too that $\Psi_{\delta}\geq 1$. We then define $\tilde{g}_{t}$ by
\ben
\tilde{g}_{t}(p)=
\left\{
\begin{array}{lcl}
g_{t}(p)  & \text{if} & p \in \Omega_{LR},\vs\\
\Psi_{\delta}\big(d(p,S_{L})\big)\, g_{t}(p) & \text{if} & p\in \Omega_{LR}^{*}\cap \Sigma^{-}_{L},\vs\\
\Psi_{\delta}\big(d(p,S_{R})\big)\, g_{t}(p) & \text{if} & p\in \Omega_{LR}^{*}\cap \Sigma^{+}_{R}
\end{array}
\right.
\een
Now chose $t^{*}$ and $\delta>0$ small enough that the boundaries of $(\Omega^{*}_{LR},\tilde{g}_{t})$ are strictly mean convex (in the outgoing directions) for any $0\leq t\leq t^{*}$. Once this is granted we can consider for every $t\leq t^{*}$ a sphere $\tilde{S}_{t}$ minimizing the $\tilde{g}_{t}$-area among all the spheres embedded in $\Omega^{*}_{LR}$ and isotopic to $S_{0}$ \cite{MR678484}. 
\begin{figure}[h]
\centering
\includegraphics[width=11cm,height=3.5cm]{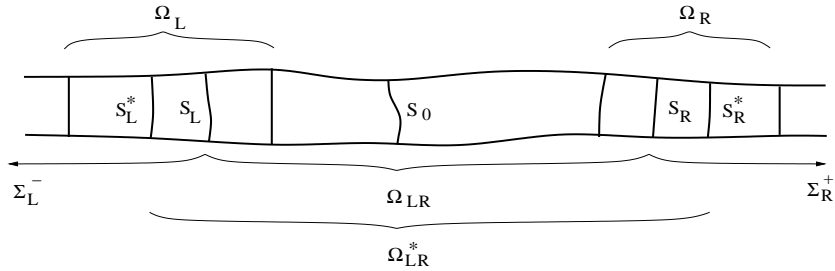}
\caption{Representation of the construction in the proof of Aux-Proposition \ref{AUXP3}.}
\label{Figure5}
\end{figure}

Until now we have not done any particular progress in the proof. The key point of the proof lies in showing that one chose $t^{*}$ smaller if necessary in such a way that the area minimizing spheres $\tilde{S}_{t}$ are embedded in ${\rm Int}(\Omega_{LR})$ and therefore do not intersect the regions where the metric $g_{t}$ was conformally modified. Once this is shown a contradiction is proved following exactly the same argument as in Proposition \ref{PAS} and will not be repeated here. 

Suppose then that there is a sequence of times $t_{i}\downarrow 0$ such that for each $t_{i}$ the minimal and stables sphere $\tilde{S}_{t_{i}}$ is not strictly embedded in ${\rm Int}(\Omega_{LR})$. Take then a subsequence (indexed again by ``i'') such that $\tilde{S}_{t_{i}}$ converges to a stable minimal sphere $\tilde{S}_{0}$ intersecting $\Omega^{*}_{LR}\setminus {\rm Int}(\Omega_{LR})$. As the $\tilde{S}_{t_{i}}$ are non contractible inside $\Sigma$ then neither is $\tilde{S}_{0}$. Moreover as $A_{\tilde{g}_{t_{i}}}(\tilde{S}_{t_{i}})\leq A_{\tilde{g}_{t}}(S_{0})<4\pi$ then $A_{\tilde{g}_{t}}(\tilde{S}_{0})\leq 4\pi$. But $A_{\tilde{g}_{0}}(\tilde{S}_{0})\geq A_{g}(\tilde{S}_{0})$ because the conformal factor is greater or equal than one. Then, the sphere $\tilde{S}_{0}\subset \Sigma$ has $A(\tilde{S}_{0})\leq 4\pi$. So it must be $A(\tilde{S}_{0})=4\pi$ by the hypothesis of Lemma \ref{PP} and by Lemma \ref{RNHG} it must be a (normalized) ERN sphere. Thus $\tilde{S}_{0}\in {\mathcal F}$. But this is a contradiction as the set $\Omega^{*}_{LR}\setminus {\rm Int}(\Omega_{LR})$ does not contain any point of $\bigcup_{S\in {\mathcal F}}\{S\}$. 
\end{proof} 

\section{Perturbations containing MOTS}\label{SECDSU}

The following proposition is direct from standard elliptic estimates and is left to the reader (use $\varphi={\rm id}$ and recall that $\bar{g}=\phi^{2} g_{0}$, $\bar{K}=0$, $\bar{E}=\hat{E}_{0}\phi^{-6}$ and $\bar{B}=0$). It says that the data sets constructed in Section \ref{SECDS} are small in the sense of Definition \ref{DPER}.

\begin{Proposition}\label{P1} Given $0<\hat{\epsilon}<1/16$ and integer $k\geq 1$ there is $\varepsilon=\varepsilon(\hat{\epsilon},k)>0$ such that for any $\hat{x}\in (-\infty,-1]$ the data set $(\bar{g},\bar{K};\bar{E},\bar{B})$ constructed in Section \ref{SECDS} out of $\hat{\epsilon}$ and $\hat{x}$, is $\varepsilon$-close in $C^{k}$ to the standard ERN$_{1}$ initial data. Moreover $\varepsilon\rightarrow 0$ if we fix $k$ and let $\hat{\epsilon}\rightarrow 0$.
\end{Proposition} 

In what follows we explain a {\it pointed convergence} that will be useful inside the proof of Proposition \ref{P2}. We keep identifying $\Sigma_{0}$ to $(-\infty,0)\times {\mathbb S}^{2}$ as we did before, in particular the factor $(-\infty,0)$ is the range of $x$. Let $\hat{x}_{i}$ be a sequence diverging to minus infinity, i.e $\lim\, \hat{x}_{i}=-\infty$ and let $s_{0}$ be a fixed point in $\mathbb{S}^{2}$. 
If we ``follow'' the ERN$_{1}$ metric $g_{0}$ around the sequence of points $\hat{x}_{i}\times s_{0}$ then, as we know, it converges to the metric $\check{g}_{0}$ of the standard ERNT$_{1}$ initial data. The standard mathematical way of saying this is that the pointed sequence 
$(\Sigma_{0};g_{0};\hat{x}_{i}\times s_{0})$ converges smoothly to $(\mathbb{R}\times \mathbb{S}^{2};\check{g}_{0},0\times s_{0})$. We write this convergence by saying that for any integers $n\geq 1$ and $k\geq 1$ we have  
\be\label{CONV1}
\lim_{i\uparrow \infty}\ \ \big\|\, \varphi^{*}_{n,i}\, g_{0} - \check{g}_{0}\, \big\|_{C^{k}_{\check{g}_{0}}([-n,n]\times \mathbb{S}^{2})}\ =\ 0,
\ee
where $\varphi_{n,i}:[-n,n]\times {\mathbb S}^{2}\rightarrow [-n+\hat{x}_{i},n+\hat{x}_{i}]\times {\mathbb S}^{2}(\subset \Sigma_{0})$ is the map $\varphi_{n,i}(x,s)=(x+\hat{x}_{i},s)$ (note that $\varphi_{n,i}(0\times s_{0})=\hat{x}_{i}\times s_{0}$ for all $i$). 
More generally, the pointed sequence of initial data $(\Sigma_{0};g_{0};E_{0};\hat{x}_{i}\times s_{0})$ converges smoothly to $(\mathbb{R}\times \mathbb{S}^{2};\check{g}_{0};\check{E}_{0};0\times s_{0})$ because in addition to (\ref{CONV1}) we have 
\be\label{CONV2}
\lim_{i\uparrow \infty}\ \ \big\|\, \varphi^{*}_{n,i}\, E_{0} - \check{E}_{0}\, \big\|_{C^{k}_{\check{g}_{0}}([-n,n]\times \mathbb{S}^{2})}\ =\ 0,
\ee
for any $n\geq 1$ and $k\geq 1$. Fix now $0<\hat{\epsilon}<1/16$ and consider the sequence of vector fields $\hat{E}_{i}:=\hat{E}_{\hat{x}_{i},\hat{\epsilon}}$ given in (\ref{EHAT}) out of $\hat{x}=\hat{x}_{i}$, $\hat{\epsilon}$ and $\omega$. In the same way as before, this sequence converges smoothly to $\hat{E}_{\infty}:=\check{E}_{0} + \hat{\lambda}_{\infty}(\star\, d\, \star\, \omega_{\infty})^{\sharp}$ where $\omega_{\infty}$ is the pull-back of $\omega$ by the map $(x,s)\rightarrow (x-2,s)$ from $[-1,1]\times \mathbb{S}^{2}$ into $[-3,-1]\times \mathbb{S}^{2}$ and $\hat{\lambda}_{\infty}$ is a constant such that 
\be\label{EET}
\sup \bigg| 1-\frac{|\hat{E}_{\infty}|^{2}_{\check{g}_{0}}}{|\check{E}_{0}|^{2}_{\check{g}_{0}}}\bigg|=\hat{\epsilon}
\ee
As before this convergence is expressed by the limit
\be\label{CONV3}
\lim_{i\uparrow \infty}\ \ \big\|\, \varphi^{*}_{n,i}\, \hat{E}_{i} - \hat{E}_{\infty}\, \big\|_{C^{k}_{\check{g}_{0}}([-n,n]\times \mathbb{S}^{2})}\ =\ 0,
\ee 
for any $n\geq 1$ and $k\geq 1$. Note that as $\omega_{\infty}$ has support in $[-1,1]\times \mathbb{S}^{2}$ then $\hat{E}_{\infty}=\check{E}_{0}$ outside $[-1,1]\times \mathbb{S}^{2}$. In particular $|\hat{E}_{\infty}|_{\check{g}_{0}}=1$ outside $[-1,1]\times \mathbb{S}^{2}$ because $|\check{E}_{0}|_{\check{g}_{0}}=1$.

Now, let $\phi_{i}$ be the sequence of conformal factors constructed in Section \ref{SECDS} out of $\hat{x}_{i}$ and the fixed $\hat{\epsilon}$. Using standard elliptic estimates and the barrier bounds (\ref{LB})-(\ref{UB}) one easily shows that the sequence $\phi_{i}$ has a subsequence (indexed again by ``$i$'') converging smoothly to a limit smooth conformal factor $\phi_{\infty}>0$. Namely,
\be\label{CONV4}
\lim_{i\uparrow \infty}\ \ \big\|\, \varphi^{*}_{n,i}\, \phi_{i} - \phi_{\infty}\, \big\|_{C^{k}_{\check{g}_{0}}([-n,n]\times \mathbb{S}^{2})}\ =\ 0,
\ee
for any $n\geq 1$ and $k\geq 1$. Moreover because of (\ref{CONV1}), (\ref{CONV2}) and (\ref{CONV3}) the limit conformal factor $\phi_{\infty}$ satisfies the limit L-equation
\be\label{LEQL}
\Delta_{\check{g}_{0}} \phi_{\infty}=2|\check{E}_{0}|_{\check{g}_{0}}^{2}\phi_{\infty} - 2|\hat{E}_{\infty}|_{\check{g}_{0}}^{2}\phi_{\infty}^{-3}.
\ee
The convergences (\ref{CONV1}), (\ref{CONV2}), (\ref{CONV3}) and (\ref{CONV4}) also show that the pointed subsequence $(\bar{g}_{i}=\phi_{i}^{4}g_{0};\bar{E}_{i}=\hat{E}_{i}=\phi^{-6}\hat{E}_{i};\hat{x}_{i}\times s_{0})$ converges smoothly to $(\bar{g}_{\infty}:=\phi^{4}_{\infty}\check{g}_{0};\bar{E}_{\infty}:=\phi^{-6}_{\infty}\hat{E}_{\infty},0\times s_{0})$. 

It is an important fact that the limit data set $(\mathbb{R}\times \mathbb{S}^{2};\bar{g}_{\infty};\bar{E}_{\infty})$ is never the ERNT$_{1}$ initial data. If this were the case then we would have $|\bar{E}_{\infty}|^{2}_{\bar{g}_{\infty}}=1$ and therefore $|\hat{E}_{\infty}|_{\check{g}_{0}}^{2}\phi^{-8}_{\infty}=1$. Plugging this in (\ref{LEQL}) and recalling that $|\check{E}_{0}|_{\check{g}_{0}}=1$ we would obtain 
\ben
\Delta_{\check{g}_{0}}\phi_{\infty}=2\phi_{\infty} - 2\phi_{\infty}^{5}.
\een
Then observe that as $|\hat{E}_{\infty}|_{\check{g}_{0}}=1$ outside $[-1,1]\times \mathbb{S}^{2}$ we would have $\phi_{\infty}=1$ also outside $[-1,1]\times \mathbb{S}^{2}$. Then, as the constant function one is a solution of (\ref{LEQL}) we must have $\phi_{\infty}=1$ everywhere by the unique continuation principle. Thus, it would be $|\hat{E}_{\infty}|_{\check{g}_{0}}=1$ everywhere, contradicting (\ref{EET}). 

Observe that any non-contractible surface $S$ embedded in $(\mathbb{R}\times \mathbb{S}^{2};\bar{g}_{\infty})$ must have $\bar{g}_{\infty}$-area greater or equal than $4\pi$. To see this use Proposition \ref{PASO} to have $A_{\bar{g}_{\infty}}(S)=\lim A_{\bar{g}_{i}}(\varphi_{n,i}(S))$ $\geq 4\pi$. Similarly we have $|Q_{\rm E}(S)|=\lim |Q_{\rm E}(\varphi_{n,i}(S))|=1$. We can now use this information together with Lemma \ref{PP} and the fact that the limits $(\bar{g}_{\infty};\bar{E}_{\infty})$ are not the ERNT$_{1}$ initial data, to conclude that for any non-contractible embedded $S$ we have $A_{\bar{g}_{\infty}}(S)>4\pi$. This will be crucially used in the following proposition.

\begin{Proposition}\label{P2} Let $0<\hat{\epsilon}<1/16$. Then there is $\hat{x}_{0}=\hat{x}_{0}(\hat{\epsilon})$ such that for any $\hat{x}\leq \hat{x}_{0}$ the data set $(\bar{g},\bar{K};\bar{E},\bar{B})$ constructed in Section \ref{SECDS} out of $\hat{\epsilon}$ and $\hat{x}$ possess an embedded minimal and stable sphere $M$ separating the two ends. Because $\bar{K}=0$ such sphere is also a MOTS (to the past and to the future).
\end{Proposition} 

\begin{proof}[\bf Proof.] We will proceed by contradiction. Assume therefore that there is $0<\hat{\epsilon}<1/16$ and a sequence $\hat{x}_{i}\rightarrow -\infty$  such that, if we denote by $((-\infty,0)\times {\mathbb S}^{2};\bar{g}_{i};\bar{E}_{i})$ the data sets constructed out of $\hat{\epsilon}$ and $\hat{x}_{i}$, then none of the manifolds $((-\infty,0)\times \mathbb{S}^{2};\bar{g}_{i})$ possess a stable minimal sphere $M$ separating the two ends. We will see that this leads to a contradiction. 

Firstly, as commented before, one can take a subsequence of the pointed sequence $((-\infty,0)\times {\mathbb S}^{2};\bar{g}_{i};\bar{E}_{i};\hat{x}_{i}\times s_{0})$ converging (in the pointed sense) to a smooth data set $((-\infty,\infty)\times {\mathbb S}^{2}; \bar{g}_{\infty},\bar{E}_{\infty})$. Moreover and as commented above, for any embedded sphere $S$ isotopic to $S_{0}:=0\times {\mathbb S}^{2}$ we have $A_{\bar{g}_{\infty}}(S)> 4\pi$.  

Secondly, let $\psi_{\delta}(z)$ be the smooth real function of the one variable $z\in [-1,\infty]$ defined as
\ben
\psi_{\delta}(z)=
\left\{
\begin{array}{lcl}
1+e^{\displaystyle 1/z+1/(z+1+\delta)} & {\rm if} & z\in [-1,0],\\
1 & {\rm if} & z\in [0,\infty)
\end{array}
\right.
\een
With this function define the metric $\tilde{g}_{i}=[\psi_{\delta}(x-\hat{x}_{i})]\, \bar{g}_{i}$  on the manifold $[-1+\hat{x}_{i},0)\times {\mathbb S}^{2}$ and set $\delta>0$ small enough that the boundary $(-1+\hat{x}_{i})\times {\mathbb S}^{2}$ of $[-1+\hat{x}_{i},0)\times {\mathbb S}^{2}$ is strictly mean convex (in the direction of decreasing $x$) for all $i$. 
Of course the pointed sequence $([-1+\hat{x}_{i},0)\times {\mathbb S}^{2}; \tilde{g}_{i}; \hat{x}_{i}\times s_{0})$ converges to $([-1,\infty)\times {\mathbb S}^{2};\tilde{g}_{\infty})$ where $\tilde{g}_{\infty}=\psi_{\delta}\, \bar{g}_{\infty}$ and because $\psi_{\delta}\geq 1$ we have $A_{\tilde{g}_{\infty}}(S)\geq A_{\bar{g}_{\infty}}(S)>4\pi$ for any embedded sphere isotopic to $S_{0}$.

Thirdly, recall that $\bar{g}_{i}=\phi^{2}_{i}\, g_{0}$ where $\phi_{i}$ is a solution to the L-equation enjoying the upper and lower bounds $\phi_{i,-}\leq \phi_{i} \leq \phi_{i,+}$ where $\phi_{i,\pm}$ are given by (\ref{LB})-(\ref{UB}) with $\hat{x}=\hat{x}_{i}$. 
In particular the conformal factor $\phi_{i}$ restricted to the spheres $S_{\hat{x}_{i}/2}:=\{x=\hat{x}_{i}/2\}$ is bounded below by $1-4\hat{\epsilon}/\cosh (\hat{x}_{i}/2)$ and above by $1+4\hat{\epsilon}/\cosh (\hat{x}/2)$. This implies that $A_{\bar{g}_{i}}(S_{\hat{x}_{i}/2})\rightarrow 4\pi$ and therefore that $A_{\tilde{g}_{i}}(S_{\hat{x}_{i}/2})\rightarrow 4\pi$. 
Let $\tilde{S}_{i}\subset [-1+\hat{x}_{i},0)\times {\mathbb S}^{2}$ be the embedded sphere minimizing the $\tilde{g}_{i}$-area among all spheres embedded in $[-1+\hat{x}_{i},0)\times \mathbb{S}^{2}$ and isotopic to $S_{0}$. Such sphere always exists because $([-1+\hat{x}_{i},0)\times {\mathbb S}^{2},\tilde{g}_{i})$ has strictly mean convex boundary and is asymptotically flat \cite{MR678484}. Moreover, as $A_{\tilde{g}_{i}}(S_{\hat{x}_{i}/2})\rightarrow 4\pi$ and as $A_{\tilde{g}_{i}}(\tilde{S}_{i})\geq 4\pi$ for all $i$ then we must have $A_{\tilde{g}_{i}}(\tilde{S}_{i})\rightarrow 4\pi$.  

On the other hand every surface $\tilde{S}_{i}$ must intersect $[-1+\hat{x}_{i},\hat{x}_{i}]\times {\mathbb S}^{2}$, which is the domain where $\tilde{g}_{i}$ differs from $\bar{g}_{i}$, otherwise $\tilde{S}_{i}$ would be $\bar{g}_{i}$-minimal and stable which is against the assumption. Now, take another subsequence if necessary in such a way that $\tilde{S}_{i}$ converges to a $\tilde{g}_{\infty}$-minimal and stable sphere intersecting $[-1,0]\times {\mathbb S}^{2}$ (inside the limit space) and isotopic to $S_{0}:=0\times \mathbb{S}^{2}$. As discussed before we must have $A_{\tilde{g}_{\infty}}(\tilde{S}_{\infty})>4\pi$ and at the same time $A_{\tilde{g}_{\infty}}(\tilde{S}_{\infty})=\lim A_{\tilde{g}_{i}}(\tilde{S}_{i})=4\pi$ which is a contradiction.
\end{proof}
  
\section{Proof of the main result}\label{PMR}

We are ready to prove the main result of this article. For the convenience of the reader we restate it below. 

\begin{figure}[h]
\centering
\includegraphics[width=7.7cm,height=5.7cm]{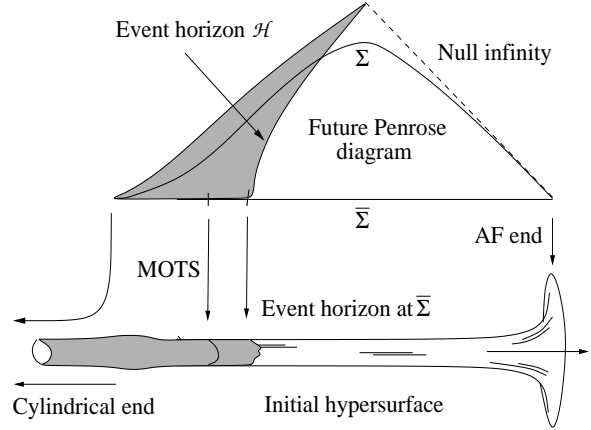}
\caption{Picture of the geometric construction in the argument by contradiction of the proof of Theorem \ref{T1}.}
\label{Figure6}
\end{figure}
\begin{Theorem}\label{T1}For any $\bar{\varepsilon}>0$ and integer $k\geq 1$ there is a smooth and maximal electro-vacuum data set $(\bar{\Sigma};\bar{g},\bar{K};\bar{E},\bar{B})$, $\bar{\varepsilon}$-close in $C^{k}$ to the standard ERN initial data and falling into it exponential along the cylindrical end, which cannot decay, towards the future or the past, into any EKN solution.
\end{Theorem}
\begin{proof}[\bf Proof.]  Set $\hat{\epsilon}$ be small enough in such a way that the $\varepsilon(\hat{\epsilon},k)$ provided by Proposition \ref{P1} is less or equal than $\bar{\varepsilon}$. Let then $\hat{x}$ be any number less or equal than the $\hat{x}(\hat{\epsilon})$ provided by Proposition \ref{P2} and let $(\bar{g},\bar{K}=0;\bar{E},\bar{B}=0)$ be the axisymmetric and time symmetric data set constructed in Section \ref{SECDS} out of $\hat{\epsilon}$ and $\hat{x}$. By Proposition \ref{P1} such data set is $\bar{\varepsilon}$-close in $C^{k}$ to the standard ERN initial data. Its total electromagnetic charges are $Q_{\rm E}=1$, $Q_{\rm M}=0$ and the total angular momentum is $J=0$.
Moreover the data set falls off exponentially towards the background data set $(g_{0},K_{0};E_{0},B_{0})$ along the cylindrical end as explained at the end of Section \ref{SECDS}.
Also, by Proposition \ref{P2}, such data set possess a stable minimal surface $M$ separating the two ends, which is therefore a future and past MOTS.  For this reason the following argument applies equally to the future and to the past. Here we will argue only to the future. The future globally hyperbolic development of the initial data will be denoted by $({\bf M}^{+};{\bf g})$.

Suppose now that the future evolution of the initial data set $(\bar{g},\bar{K};\bar{E},\bar{B})$ decays into a EKN space-time. In such case $M$ acts as a barrier preventing the event horizon ${\mathcal H}$ to enter the region in $\bar{\Sigma}$ enclosed between $M$ and the cylindrical end \footnote{Because of the presence of $M$ the space-time ${\bf M}^{+}$ must have a horizon, namely $\partial (J^{-}({\mathscr S}^{+})\cap ({\bf M}^{+} \setminus \bar{\Sigma})) \neq \emptyset$. Of course we assume the existence of a Scri as in \cite{Chrusciel:2000cu} to ensure the monotonicity of the horizon's areas.} (see Figure \ref{Figure6}). In particular the intersection $\bar{H}$ between ${\mathcal H}$ and the initial hypersurface $\bar{\Sigma}$ is a compact set in $\bar{\Sigma}$ separating its two ends. 

As proved in \cite{Chrusciel:2000cu} ({\bf Proposition 3.4}) the intersection $H={\mathscr H}\cap \Sigma$ between the event horizon ${\mathscr H}$ and a Cauchy hypersurface $\Sigma$ is a two-rectifiable set of well defined area (${\mathcal H}^{2}$-Hausdorff measure). Moreover for any two Cauchy hypersurfaces $\Sigma_{1}$ and $\Sigma_{2}$, with $\Sigma_{2}$ strictly to the future of $\Sigma_{1}$, we have $A(H_{2})\geq A(H_{1})$ ($H_{i}={\mathscr H}\cap \Sigma_{i},\ i=1,2$) and if equality holds then the part of ${\mathscr H}$  between $\Sigma_{1}$ and $\Sigma_{2}$ is smooth ({\bf Theorem 6.1} in \cite{Chrusciel:2000cu}).
This monotonicity allows us to define the ``future limit of the areas of the horizon's sections'', denoted here by $\lim_{\Sigma\uparrow} A({\mathcal H}\cap \Sigma)$, in the following simple manner. Take any sequence of Cauchy hypersurfaces $\Sigma_{i}$ such that, (i) $\Sigma_{i'}$ lies strictly to the future of $\Sigma_{i}$ when $i'>i$, and (ii) for any $p\in {\bf M}^{+}$ there is $i(p)$ such that for all $i\geq i(p)$ the point $p$ does not lie in the future of $\Sigma_{i}$. Then, define
\ben
\lim_{\Sigma\uparrow}\, A({\mathcal H}\cap \Sigma):=\lim_{i\uparrow \infty}\, A({\mathcal H}\cap \Sigma_{i}).
\een
It is easily checked that this definition does not dependent on the sequence $\Sigma_{i}$.

Now, if the future evolution of the initial data decay into an extreme EKN solution, then, as the electromagnetic charges and the angular momentum are conserved, the EKN limit must necessarily be ERN$_{1}$ and we must have $\lim_{\Sigma\uparrow}\, A({\mathcal H}\cap \Sigma)=4\pi$. Hence $A(\bar{H})\leq \lim_{\Sigma\uparrow} A({\mathcal H}\cap \Sigma)=4\pi$. If $A(\bar{H})=4\pi$ then $4\pi=A({\mathcal H}\cap \Sigma)$ for all $\Sigma$ stricly in the future of $\bar{\Sigma}$ and the whole ${\mathscr H}$ must be smooth. This implies that $\bar{H}=\mathcal{H}\cap \bar{\Sigma}$ is also smooth because ${\mathcal H}$ and $\bar{\Sigma}$ intersect transversely \footnote{In principle $\bar{H}$ can have several connected components, but at least one of them must separate the two ends of $\bar{\Sigma}$.}. Proposition \ref{PAS} then tells us that $A(\bar{H})>4\pi$ and we reach a contradiction. Hence it must be $A(\bar{H})<4\pi$. 
On the other hand the initial hypersurface $(\bar{\Sigma},\bar{g})$ has one asymptotically flat end and one cylindrical end asymptotic to the metric product of $\mathbb{R}$ and the unit two-sphere $\mathbb{S}^{2}$ (which has area $4\pi$). On these grounds and based on general results of geometric measure  theory \cite{MR756417}, (see also \cite{MR678484}), we can guarantee the existence of a smooth area-minimizer in the class of compact two-rectifiable sets separating the two ends. Such minimizer must have area less than $4\pi$ because $A(\bar{H})<4\pi$ and because $\bar{H}$ is rectifiable and separating. By Proposition \ref{PAS} the area of the smooth minimizer must be greater than $4\pi$ and we reach again a contradiction. It follows that the future evolution of the initial data cannot decay into a EKN solution.
\end{proof}

\section{Acknowledgment} We would like to thank Sergio Dain for useful suggestions in the early stages of this article. 

\bibliographystyle{plain}
\bibliography{Master}

\end{document}